\documentclass[12pt]{article}
\usepackage{amsmath}
\usepackage{amssymb}
\usepackage{amsthm}
\usepackage{authblk}
\usepackage[margin=1.2in]{geometry}
\usepackage{microtype}
\usepackage{dsfont}
\usepackage{hyperref}
\usepackage{graphicx}
\usepackage{subcaption}

\theoremstyle{plain}
\theoremstyle{definition}
\newtheorem{definition}{Definition}
\newtheorem{theorem}{Theorem}
\newtheorem{lemma}{Lemma}
\newtheorem{assumption}{Assumption}
 
\newtheorem{remark}{Remark}
\newtheorem{proposition}{Proposition}

\newcommand{\x}{\pmb{x}}
\newcommand{\y}{\pmb{y}}
\newcommand{\D}{\pmb{D}}
\newcommand{\I}{\pmb{I}}
\newcommand{\X}{\pmb{X}}
\newcommand{\Y}{\pmb{Y}}
\newcommand{\Z}{\pmb{Z}}
\newcommand{\xu}{\pmb{u}}
\newcommand{\xv}{\pmb{v}}
\newcommand{\vz}{\pmb{z}}
\newcommand{\ve}{\pmb{e}}
\newcommand{\vf}{\pmb{f}}
\newcommand{\px}{\pmb{\xi}}
\newcommand{\Q}{\pmb{Q}}
\newcommand{\W}{\pmb{W}}
\newcommand{\Fm}{\pmb{F}}
\newcommand{\ind}{\mathds{1}}
\newcommand{\tr}{\operatorname{tr}}
\newcommand{\E}{\mathbb{E}}
\newcommand{\C}{\mathbb{C}}
\newcommand{\mA}{\mathcal{A}}

\def\lV{\left\lVert}
\def\rV{\right\lVert}
\def\lv{\left\lvert}
\def\rv{\right\lvert}
\def\lk{\left(}
\def\rk{\right)}
\def\lg{\langle}
\def\rg{\rangle}
\def\lz{\left[}
\def\rz{\right]}

\begin{document}
\title{PhaseLift for Coded Diffraction Patterns: Optimal Sampling Rate}

\author{Gao Huang\footnote{School of Mathematical Sciences, Zhejiang University, Hangzhou 310027, P. R. China, E-mail address: hgmath@zju.edu.cn} }	
\author{Song Li\footnote{School of Mathematical Sciences, Zhejiang University, Hangzhou 310027, P. R. China, E-mail address: songli@zju.edu.cn}}
\date{\today}
\affil{}
\renewcommand*{\Affilfont}{\small\itshape}

\maketitle

\begin{abstract}
Recovering a complex-valued signal from coded diffraction patterns, namely the Fourier intensities obtained after modulating the signal with a collection of masks, is a fundamental structured phase retrieval problem arising in diffraction imaging and related applications.
Despite its practical importance, the theoretical analysis of this structured framework remains scarce.
In the standard random mask model, the optimal sampling rate achievable by computationally tractable recovery methods has remained open.

In this paper, we establish the optimal sampling rate for the PhaseLift feasibility program.
More precisely, PhaseLift achieves exact recovery of an unknown signal $\x_0\in\C^n$, up to a global phase,
from $\mathcal{O}\lk\log n\rk$ random masks, with polynomially decaying failure probability.
Since $\Omega\lk\log n\rk$ masks are necessary to identify certain signals under the erasure mask ensemble, our result thereby achieves the optimal mask complexity.
Equivalently, PhaseLift attains the optimal total sampling rate of $m=\mathcal{O}\lk n\log n\rk$ scalar intensity measurements.
The proof is based on an approximate dual certificate construction via a refined golfing scheme that combines adaptive mask allocation with a dimension-independent truncation threshold.
\end{abstract}



\section{Introduction}

Many imaging systems record only the intensity of a wave field, while its phase is lost.
Under the Fraunhofer diffraction approximation, the complex field at the detector is given by the Fourier transform of the object, whereas only its squared magnitude is observed.
This gives rise to the Fourier phase retrieval problem~\cite{bendory2017fourier}:
reconstruct an unknown signal $\x_0\in\C^n$ from the intensity measurements
\begin{equation}\label{eq:fourier}
        y_k=\lv\sum_{t=0}^{n-1} \x_0(t) e^{-\frac{2\pi\mathrm{i}kt}{N}}\rv^2,
        \qquad 0\le k\le N-1.
\end{equation}
Taking $N=2n-1$ yields an oversampled Fourier acquisition that determines the complete aperiodic autocorrelation of $\x_0$.
Fourier phase retrieval arises naturally in optics and crystallography and remains fundamental in coherent diffraction imaging, X-ray crystallography, astronomical imaging, and microscopy~\cite{fienup1987phase,millane1990phase,harrison1993phase,shechtman2015overview}.

The difficulty, however, extends well beyond the quadratic nonlinearity of the measurement map.
Even when the measurements in~\eqref{eq:fourier} determine the complete aperiodic autocorrelation of $\x_0$, they generally do not determine the signal uniquely~\cite{bendory2017fourier}. 
In one dimension, reciprocal conjugate pairs of roots of the autocorrelation polynomial can generate exponentially many signals with identical Fourier magnitudes, beyond the unavoidable ambiguities of global phase, translation, and conjugate reflection~\cite{beinert2015ambiguities,bendory2017fourier}. 
Moreover, uniqueness alone provides neither quantitative stability nor a provably efficient reconstruction procedure, and classical alternating projection methods lack general global convergence guarantees~\cite{gerchberg1972practical,fienup1982algorithms,bendory2017fourier,fannjiang2020numerics}. 
These challenges in identifiability, stability, and computation motivate the introduction of additional measurement diversity through multiple masks or coded illuminations~\cite{bandeira2014masked,candes2015cdp}.

Coded diffraction pattern (CDP) measurements provide a physically natural way to introduce such diversity while preserving the Fourier structure of the acquisition. 
The object is modulated by $L$ known masks, and the Fourier intensity of each masked signal is recorded:
\begin{equation}\label{eq:cdp0}
y_{k,\ell}=\lv\lg\vf_k,\D_\ell\x_0\rg\rv^2,
\qquad
1\le k\le n,\quad 1\le \ell\le L,
\end{equation}
where $\D_\ell=\operatorname{diag}\lk\pmb{d}_\ell\rk$ denotes the diagonal modulation associated with the $\ell$-th mask and $\{\vf_k\}_{k=1}^n$ are the columns of the $n\times n$ discrete Fourier transform matrix.
Such measurements can be implemented by modulating the incident illumination or by placing a mask or phase plate after the object~\cite{fannjiang2012phase,loewen2018diffraction}. 
Since each mask produces a complete diffraction pattern, the natural sampling parameter is the number $L$ of masks, corresponding to a total of $m:=nL$ scalar intensity measurements.

In this paper, we focus on random modulation patterns, where the masks $\left\{\D_\ell\right\}_{\ell=1}^L$ are generated according to a prescribed entrywise random model~\cite{candes2015cdp,candes2015wirtinger,gross2017improved}.
Randomization provides a flexible way to generate diverse illuminations while avoiding the more delicate task of designing deterministic masks with general recovery guarantees.
Even so, the random CDP model remains highly structured: 
%substantially more structured than standard random ensembles consideredin~\cite{candes2013phaselift,candes2014solving,krahmer2020complex,huang2025stable}: 
within each diffraction pattern, all Fourier frequencies are coupled through a common random mask.
More precisely, for each fixed mask $\D_\ell$, the vectors $\left\{\D_\ell^*\vf_k\right\}_{k=1}^n$ share the same random mask entries and are therefore dependent across Fourier frequencies.
This within-pattern dependence, combined with the rigid Fourier structure, constitutes the main technical challenge in the analysis.

For independent Gaussian, sub-Gaussian, and heavy-tailed sampling ensembles, a mature theory has been developed: convex lifting methods achieve optimal or near-optimal guarantees for exact recovery, stability, and robustness~\cite{candes2013matrix,candes2013phaselift,candes2014solving,chen2015exact,cai2015rop,kueng2017low,kabanava2016stable,krahmer2020complex,huang2025robust,huang2025stable,huang2026low}, while computationally efficient nonconvex algorithms attain comparable sampling rates with provable global convergence guarantees~\cite{candes2015wirtinger,chen2017solving,wang2017solving,zhang2017nonconvex,duchi2019solving,chen2019gradient,tan2023online,godeme2023provable}. 
However, the corresponding theory for CDP phase retrieval under random masks remains considerably scarce.
In particular, Cand\`es, Li, and Soltanolkotabi raised the problem of determining the optimal number of masks, and hence the optimal sampling rate, for computationally tractable recovery~\cite{candes2015cdp,candes2015wirtinger}, whose answer still remained elusive.
They proved that the PhaseLift feasibility program recovers any fixed complex signal from $\mathcal{O}\lk\log^4 n\rk$ masks~\cite{candes2015cdp} and they also established geometric convergence of Wirtinger Flow with the same mask order~\cite{candes2015wirtinger}. 
Gross, Krahmer, and Kueng subsequently reduced the sufficient number of masks for PhaseLift to $\mathcal{O}\lk\log^2 n\rk$ and they also showed that, for the erasure mask distribution, any recovery method requires $\Omega\lk\log n\rk$ masks to identify even certain coordinate signals~\cite{gross2017improved}.
Thus, prior to the present work, a logarithmic gap remained between the best computational guarantee for exact recovery and the corresponding lower bound.

More recent nonconvex results for CDP phase retrieval include resampled methods that recover to accuracy $\varepsilon$ with $\varepsilon$-dependent sample rate~\cite{li2022sampling,li2025taf,gao2025affine}, as well as a local linear convergence guarantee for mirror descent under a substantially larger than logarithmic mask requirement~\cite{godeme2023provable}.
The former use fresh masks at each iteration and thus do not establish exact recovery from a fixed collection of masks, 
whereas the latter guarantees only local convergence rather than global exact recovery.
Several works have also investigated the robustness of recovery under the random mask model to measurement noise~\cite{soltanolkotabi2014algorithms,li2021masks,huang2026recovery}.
A complementary literature investigates alternative mask designs, acquisition schemes, uniqueness conditions, and computational frameworks beyond the random mask model~\cite{bandeira2014masked,jaganathan2015phase,chen2018fourier,chen2018phase,guerrero2020phase,jaming2023uniqueness,ye2024green,wang2026photonic}.
Although these developments provide important guarantees in related settings, the logarithmic gap under the random mask model persists.

Motivated by the open problem raised by Cand\`es, Li, and Soltanolkotabi~\cite{candes2015cdp,candes2015wirtinger}, together with the logarithmic lower bound established by Gross, Krahmer, and Kueng~\cite{gross2017improved}, we are led to the following question:
\begin{quote}
\textbf{\textit{Can computationally tractable recovery for CDP phase retrieval~\eqref{eq:cdp0} be achieved using only $\mathcal{O}\lk\log n\rk$ random masks, which is optimal in order?}}
\end{quote}

In this paper, we answer this question affirmatively.
We prove that, for any unknown signal $\x_0\in\C^n$ and $\omega\ge1$, the PhaseLift feasibility program exactly recovers $\x_0$, up to a global phase, from $L\ge C\omega\log n$ random masks with polynomially decaying failure probability $n^{-\omega}$,
where $C$ depends only on the mask distribution parameters.
Together with the information-theoretic lower bound, our result is optimal in its joint dependence on the dimension and the failure probability.
Equivalently, PhaseLift attains the optimal total sampling rate $m=\mathcal{O}\lk\omega n\log n\rk$, and hence $m=\mathcal{O}\lk n\log n\rk$ for fixed $\omega$.
%This improves upon the previous PhaseLift guarantees in~\cite{candes2015cdp,gross2017improved}.
%Consequently, PhaseLift achieves the optimal mask order $\mathcal{O}\lk\log n\rk$ and the corresponding optimal total sampling rate %$m=\mathcal{O}\lk n\log n\rk$, matching the information-theoretic lower bound. % established in~\cite{gross2017improved}.

The proof proceeds by constructing an approximate dual certificate via a refinement of the golfing schemes developed in~\cite{candes2015cdp,gross2017improved}.
Alternative approaches to establishing PhaseLift recovery guarantees, such as those based on rank-one RIP-$\ell_1/\ell_2$ estimates~\cite{chen2015exact,cai2015rop} and small ball method~\cite{kueng2017low,krahmer2020complex,huang2025low,huang2026low},
are less suited to the CDP setting, due to its rigid Fourier structure and strong within-pattern dependence.
Our argument instead exploits the conditional structure inherent in the golfing construction and introduces two main refinements.
First, the masks are allocated adaptively across the golfing steps, with the batch sizes decreasing as the residual contracts, in contrast to the fixed batch sizes used in~\cite{candes2015cdp} and the two stage allocation scheme of~\cite{gross2017improved}.
Second, each golfing update is constructed from an adaptively truncated operator with a dimension-independent truncation threshold.
We believe that these ideas are of independent interest and may prove useful in the analysis of other recovery methods with CDP measurements, as well as inverse problems involving structured sampling ensembles.

\subsection{Problem Setup}\label{sec:setup}

\paragraph{CDP Measurements.}
Let $\x_0\in\C^n$ be an unknown signal. 
Given $L$ masks $\D_1,\ldots,\D_L$, the CDP measurements are given by
\begin{equation*}\label{eq:cdp}
        y_{k,\ell}=\lv\lg \vf_k,\D_\ell \x_0\rg\rv^2,\qquad
        1\le k\le n,\quad 1\le \ell\le L.
\end{equation*}
We lift the signal to the rank-one PSD matrix $\X_0:=\x_0\x_0^*$.
For $1\le k\le n$ and $1\le \ell\le L$, define the measurement matrices $\Fm_{k,\ell}:=\D_\ell^*\vf_k\vf_k^*\D_\ell$.
Then the measurements can be written linearly in the lifted variable:
\begin{equation*}
        y_{k,\ell}=\tr\lk\Fm_{k,\ell}\X_0\rk,\qquad
        1\le k\le n,\quad 1\le \ell\le L.
\end{equation*}
Let $\mathcal{H}_n$ denote the space of $n\times n$ Hermitian matrices, equipped with the Hilbert--Schmidt inner product
$\lg\Z,\W\rg=\tr\lk\Z\W\rk$.
The associated linear sampling map $\mA:\mathcal{H}_n\to\mathbb{R}^{nL}$ is defined componentwise by
\begin{equation}\label{eq:sampling}
        \mA\lk\Z\rk :=\left\{\tr\lk\Fm_{k,\ell}\Z\rk\right\}_{1\le k\le n,1\le \ell\le L}.
\end{equation}

\paragraph{Random Mask Model.} 
We adopt the random mask model considered in~\cite{candes2015cdp,candes2015wirtinger,gross2017improved}. 

\begin{assumption}\label{ass:mask}
Let $\pmb{d}\in\C^n$ be a random vector whose entries are i.i.d.\ copies of a complex-valued random variable $d$, and let
$\pmb{d}_1,\ldots,\pmb{d}_L$ be independent copies of $\pmb{d}$.
For each $1\le \ell\le L$, set $\D_\ell=\operatorname{diag}\lk\pmb{d}_\ell\rk$.
The distribution of $d$ is symmetric and satisfies
\begin{equation*}
    \E d=0, \qquad \E d^2=0,
    \qquad
    \E\lv d\rv^4= 2\lk\E\lv d\rv^2\rk^2,
    \qquad
    \lv d\rv\le M
    \quad\text{almost surely},
\end{equation*}
for some fixed constant $M>0$.
We denote $\nu:=\E\lv d\rv^2>0$.
When $n$ is odd, the condition $\E d^2=0$ may be omitted.
\end{assumption}

\begin{remark}\label{rem:mask}
The random mask model in Assumption~\ref{ass:mask} originates from~\cite{candes2015cdp}.
A related model for real-valued masks was considered in~\cite{gross2017improved}, where the analysis requires $n$ to be odd. 
Furthermore, when $n$ is odd, the condition $\E d^2=0$ in Assumption~\ref{ass:mask} can be omitted; see~\cite[Section 1.5]{candes2015cdp}. This permits the use of real-valued masks.
\end{remark}

An example satisfying Assumption~\ref{ass:mask} is the octanary mask ensemble $d=b_1b_2$ introduced in~\cite{candes2015cdp}, where $b_1$ and $b_2$ are independent random variables distributed as
\begin{equation}\label{eq:octanary_mask}
    b_1\sim
    \begin{cases}
        1,  & \text{with probability }1/4,\\
        -1, & \text{with probability }1/4,\\
        -\mathrm{i}, & \text{with probability }1/4,\\
        \mathrm{i},  & \text{with probability }1/4,
    \end{cases}
    \qquad
    b_2\sim
    \begin{cases}
        1/\sqrt{2}, & \text{with probability }4/5,\\
        \sqrt{3},   & \text{with probability }1/5.
    \end{cases}
\end{equation}
For this ensemble, $M=\sqrt{3}$ and $\nu=1$.
When $n$ is odd, Assumption~\ref{ass:mask} includes, in particular, the real-valued erasure mask ensemble considered in~\cite{candes2015cdp,gross2017improved}, given by
\begin{equation}\label{eq:erasure_mask}
        d\sim
        \begin{cases}
        \sqrt{2}, & \text{with probability }1/4,\\
        0,        & \text{with probability }1/2,\\
        -\sqrt{2},& \text{with probability }1/4.
        \end{cases}
\end{equation}
For this ensemble, one has $M=\sqrt{2}$ and $\nu=1$.

\paragraph{PhaseLift Program.} 
The PhaseLift approach was introduced by Candès et al.~\cite{candes2013matrix,candes2013phaselift} and has since been developed in several variants.
We adopt the feasibility formulation considered in~\cite{gross2017improved}, in which the signal energy $\lV\x_0\rV_2^2$ is available as part of the problem data.
This quantity can be obtained, for instance, by using one additional deterministic calibration mask $\D_0=\I$.
Indeed, Fourier orthogonality yields $\sum_{k=1}^n \lv\lg \vf_k,\x_0\rg\rv^2 = n\lV\x_0\rV_2^2$.
Accordingly, we impose the trace constraint $\tr\lk\Z\rk=\tr\lk\X_0\rk$ in the feasibility program.
On the dual side, incorporating the trace constraint allows the dual certificate to contain a multiple of the identity and hence to lie in $\operatorname{range}\lk\mathcal{A}^*\rk+\operatorname{span}\left\{\I\right\}$.
By homogeneity, we normalize $\lV\x_0\rV_2=1$ in the sequel. 
Under this normalization, the PhaseLift feasibility problem considered in this paper is
\begin{equation}\label{eq:PhaseLift}
\begin{array}{ll}
\text{find} & \Z\in\mathcal{H}_n,\\[2mm]
\text{subject to}
& \mA\lk\Z\rk=\mA\lk\X_0\rk,\\
& \Z\succeq 0,\\
& \tr\lk\Z\rk=1.
\end{array}
\end{equation}

\subsection{Main Theorem}\label{sec:theorem}

We now state our main recovery guarantee: 
$\mathcal{O}\lk\log n\rk$ random masks suffice for exact recovery via the PhaseLift feasibility program~\eqref{eq:PhaseLift}.

\begin{theorem}\label{thm:main}
Let $n\ge2$ and let $\x_0\in\C^n$ be an unknown signal satisfying $\lV\x_0\rV_2=1$, and suppose that the CDP measurements in~\eqref{eq:cdp0} are generated according to the random mask model in Assumption~\ref{ass:mask} with mask parameters $M$ and $\nu$.
Then there exists a constant $C=C\lk M,\nu\rk>0$ such that, for every $\omega\ge 1$, if
\begin{equation*}
        L \ge C\,\omega\,\log n,
\end{equation*}
the matrix $\X_0=\x_0\x_0^*$ is the unique feasible point of the PhaseLift feasibility program~\eqref{eq:PhaseLift} with probability at least $1-n^{-\omega}$. 
Consequently, $\x_0$ is recovered up to a global phase.
%Here, $M$ and $\nu$ are the mask parameters specified in Assumption~\ref{ass:mask}.
\end{theorem}

We record several remarks on Theorem~\ref{thm:main}, compare it with previous recovery guarantees, and discuss its implications.

\begin{enumerate}
\item \textbf{Previous Results:}
The PhaseLift recovery guarantee in~\cite{candes2015cdp} requires $L\gtrsim_{M,\nu}\omega\log^4 n$ random masks and succeeds with polynomially decaying failure probability $1-n^{-\omega}$.
The result in~\cite{gross2017improved} reduces the PhaseLift mask requirement to $L\gtrsim_{M,\nu}\eta\log^2 n$ and guarantees exact recovery with probability at least $1-e^{-\eta}$, which is a constant success probability when $\eta$ is fixed.

\item \textbf{Mask Lower Bound:}
For the erasure mask ensemble,~\cite[Lemma~19]{gross2017improved} shows, through a failure of injectivity for coordinate signals, that any recovery algorithm requires $\Omega\lk\log n\rk$ masks.
In Appendix~\ref{ap:flat}, we show that this obstruction persists for the flat signal $\x_0=\frac{1}{\sqrt{n}}\lk1,\ldots,1\rk^\top$.
More precisely, for $\x_0$ to be uniquely identifiable with probability at least $1-n^{-\omega}$, it is necessary that
$L\gtrsim\omega\log_2 n$ with $\omega\ge1$.
Thus, this information-theoretic obstruction is not confined to sparse or coordinate signals.

\item \textbf{Optimal Sampling Rate:}
Theorem~\ref{thm:main} further reduces the sufficient number of masks to $L\gtrsim_{M,\nu}\omega\log n$, while guaranteeing exact recovery with probability at least $1-n^{-\omega}$.
Together with the preceding lower bound, this shows that the mask requirement is \textit{optimal in its joint dependence on the dimension and the prescribed failure probability}, up to constants depending on the mask distribution.
Since each mask produces $n$ scalar intensity measurements, the corresponding total number of measurements satisfies
\begin{equation*}
m=nL\gtrsim_{M,\nu} \omega n\log n.
\end{equation*}
Thus, to achieve success probability $1-n^{-\omega}$, Theorem~\ref{thm:main} attains the optimal sampling rate $m\asymp \omega n\log n$, and hence $m\asymp n\log n$ for fixed $\omega$.
If the deterministic calibration mask $\D_0=\I$, used to obtain the signal energy, is also counted, the total sampling rate becomes $\lk L+1\rk n$, which remains of the same order.

\item \textbf{Comparison at Success Probability:}
%Theorem~\ref{thm:main} also improves the dependence of the mask requirement on the target failure probability compared with~\cite{gross2017improved}.
%Indeed, 
To obtain a success probability of at least $1-n^{-\omega}$ from the result in~\cite{gross2017improved}, one must set $\eta=\omega\log n$, leading to the requirement $L\gtrsim_{M,\nu}\omega\log^3 n$.
By contrast, Theorem~\ref{thm:main} achieves the same success probability with $L\gtrsim_{M,\nu}\omega\log n$ random masks, improving the previous mask bound by a factor of $\log^2 n$.

\item \textbf{Parameter Dependence:}
The constant $C$ in Theorem~\ref{thm:main} depends only on the mask distribution parameters $M$ and $\nu$ and, after enlarging the absolute constant $\widetilde{C}>0$ if necessary, can be taken as
$\widetilde{C}\frac{M^8}{\nu^4} \lz 1+ \log^2\lk \frac{M^4}{\nu^2}\rk \rz$.
%However, we do not know whether this dependence on $M$ and $\nu$ is optimal.

\item \textbf{Nonuniform Guarantee:}
As in~\cite{candes2015cdp,gross2017improved}, Theorem~\ref{thm:main} provides a recovery guarantee for each fixed signal $\x_0$.
It does not ensure simultaneous recovery of all signals in $\C^n$ from a single realization of the random masks.
Establishing a corresponding uniform recovery guarantee remains an interesting direction for future work.
\end{enumerate}

\section{Proofs}\label{sec:proof}

This section contains the proofs of our main results. 
We first introduce some notation and basic facts that will be used throughout the analysis.
Vectors and matrices are denoted by bold lowercase and uppercase symbols, respectively, such as $\vz,\x_0,\vf_k$ and $\Z,\X_0,\Fm_{k,\ell}$.
For a vector $\vz$, we denote its Euclidean norm by $\lV\vz\rV_2$.
For a matrix $\Z$, we denote its Frobenius, operator, and trace norms by $\lV\Z\rV_F$, $\lV\Z\rV_{\mathrm{op}}$, and $\lV\Z\rV_*$, respectively.
For nonnegative quantities $a$ and $b$, we write $a\lesssim b$ if $a\le Cb$ for some absolute constant $C>0$, and define $a\gtrsim b$ analogously.
We write $a\lesssim_p b$ if the implicit constant depends only on the parameter $p$, and use $a\gtrsim_p b$ with the analogous meaning.
The notations $b_t=\mathcal{O}\lk a_t\rk$ and $b_t=\Omega\lk a_t\rk$ are equivalent to $b_t\lesssim a_t$ and $b_t\gtrsim a_t$, respectively, with the implicit constants uniform in $t$.

The tangent space associated with the rank-one matrix $\X_0=\x_0\x_0^*$ is
\begin{equation*}
        T:=
        \left\{\x_0\vz^*+\vz\x_0^*:\vz\in\C^n\right\}\subset \mathcal{H}_n.
\end{equation*}
Let $\pmb{P}_T$ and $\pmb{P}_{T^\perp}$ denote the orthogonal projections onto $T$ and $T^\perp$, respectively.
For any $\Z\in\mathcal{H}_n$ we write
\begin{equation*}
\Z_T:=\pmb{P}_T\Z, \qquad \Z_{T^\perp}:=\pmb{P}_{T^\perp}\Z. 
\end{equation*}
Under the normalization $\lV\x_0\rV_2=1$, the matrix $\X_0=\x_0\x_0^*$ is the orthogonal projector onto $\operatorname{span}\{\x_0\}$.
Hence, for every $\Z\in\mathcal{H}_n$, 
\begin{equation}\label{eq:proj0} 
\left\{ \begin{aligned} 
&\Z_T= \X_0\Z+\Z\X_0-\X_0\Z\X_0,\\[4pt] 
&\Z_{T^\perp} = \lk\I-\X_0\rk\Z\lk\I-\X_0\rk. 
\end{aligned} 
\right. 
\end{equation}

\subsection{Preliminaries}\label{subsec:pre}

We gather the preliminary ingredients required for the proof.
They include the sub-Gaussianity of masked Fourier forms, the near isotropicity of the measurement operator, the Bernstein-type inequalities used for the concentration estimates, and the robust injectivity and dual certificate conditions that ensure exact recovery from an approximate dual certificate.

\subsubsection{Sub-Gaussianity of Masked Fourier Forms}\label{subsec:subgaussian}

We first record a sub-Gaussian estimate for the masked Fourier linear forms.

\begin{lemma}\label{lm:subgaussian}
Let $\pmb{d}=\lk d_1,\ldots,d_n\rk\in\C^n$ be a random vector whose entries are independent, mean-zero, and satisfy $\lv d_j\rv\le M$ almost surely, and let $\D=\operatorname{diag}\lk\pmb{d}\rk$.
Then for every fixed $\xv\in\C^n$, every $1\le k\le n$, every $t\ge0$,
\begin{equation}\label{eq:sub_P}
        \mathbb{P}\lk\lv\lg \vf_k,\D\xv\rg\rv\ge tM\lV\xv\rV_2\rk
        \le
        4\exp\lk-t^2/4\rk.
\end{equation}
Consequently, for every $p\ge2$,
\begin{equation}\label{eq:sub_E}
        \lk\E\lv\lg \vf_k,\D\xv\rg\rv^p\rk^{1/p}
        \le 2\sqrt{2}M\sqrt p\,\lV\xv\rV_2.
\end{equation}

\end{lemma}

\begin{proof}
Write $f_{k,j}:=\lk \vf_k\rk_j$ and set $S:=\lg \vf_k,\D\xv\rg=\sum_{j=1}^n d_j \overline{f_{k,j}} v_j$.
Since $\lv f_{k,j}\rv=1$ for every $j$, the summands in $S$ are independent and centered, and their real and imaginary parts are bounded in absolute value by $M\lv v_j\rv$.
Hence, Hoeffding's inequality in the form of~\cite[Theorem~2.2.6]{vershynin2018highdimensional} gives, for every $u\ge0$,
\begin{equation*}
\mathbb{P}\lk\lv\operatorname{Re}\lk S\rk\rv\ge u\rk
\le
2\exp\lk-\frac{u^2}{2M^2\lV\xv\rV_2^2}\rk,
\end{equation*}
and the same bound holds for $\operatorname{Im}\lk S\rk$.
Since $\lv S\rv\le \sqrt{2}\max\left\{\lv\operatorname{Re}\lk S\rk\rv,\lv\operatorname{Im}\lk S\rk\rv\right\}$, taking a union bound yields~\eqref{eq:sub_P}. 
The moment estimate~\eqref{eq:sub_E} follows from the equivalence between sub-Gaussian tail bounds and moment growth, we omit it here.
\end{proof}

\subsubsection{Near Isotropicity}

We define the measurement operator associated with all $L$ masks:
\begin{equation}\label{eq:R_operator}
\mathcal{R}
:=\frac{1}{\nu^2nL}\mA^*\mA,
\qquad
\mathcal{R}\lk\Z\rk
=\frac{1}{\nu^2nL}\sum_{\ell=1}^L \sum_{k=1}^n \Fm_{k,\ell}\tr\lk\Fm_{k,\ell}\Z\rk.
\end{equation}
For a given mask, we define the corresponding block operator by
\begin{equation}\label{eq:F_operator}
\mathcal{F}\lk\Z\rk
:=
\frac{1}{\nu^2n} \sum_{k=1}^n \Fm_k\tr\lk\Fm_k\Z\rk, 
\qquad
\Fm_k:=\D^*\vf_k\vf_k^*\D,
\end{equation}
where $\D$ has the same distribution as each $\D_\ell$.

The following proposition records the near isotropicity identity under the above normalization.

\begin{proposition}[{\cite[Lemma~3.1]{candes2015cdp}},{\cite[Lemma~7]{gross2017improved}}]\label{prop:isotropicity}
Suppose that the random masks $\D_1,\ldots,\D_L$ are generated according to Assumption~\ref{ass:mask}, and let $\D$ be an independent mask with the same distribution.
Then, for every $\Z\in\mathcal{H}_n$,
\begin{equation*}
\E\mathcal{F}\lk\Z\rk
=
\E\mathcal{R}\lk\Z\rk
=
\Z+\tr\lk\Z\rk\I.
\end{equation*}
\end{proposition}

\begin{remark}\label{rem:isotropicity}
The near isotropicity identity in Proposition~\ref{prop:isotropicity} is guaranteed by the conditions in Assumption~\ref{ass:mask}.
If the condition $\E d^2=0$ is omitted, the identity still holds whenever $n$ is odd.
\end{remark}

\subsubsection{Concentration Inequalities}

We shall use the following two Bernstein-type inequalities.
The first one is the well-known matrix Bernstein inequality, which controls sums of independent self-adjoint random matrices.

\begin{lemma}[{\cite[Theorem~1.6]{tropp2012user}}]\label{lm:matrix_bernstein}
Let $\pmb{S}_1,\ldots,\pmb{S}_N$ be independent mean-zero self-adjoint random matrices of dimension $n$. Suppose that $\lV\pmb{S}_i\rV_{\mathrm{op}}\le R_1$ almost surely for some finite constant $R_1$ and set
$\sigma_1^2\ge\lV\sum_{i=1}^N\E\pmb{S}_i^2\rV_{\mathrm{op}}$.
Then, for every $t\ge0$,
\begin{equation*}
        \mathbb{P}\lk\lV\sum_{i=1}^N\pmb{S}_i\rV_{\mathrm{op}}\ge t\rk
        \le
        2n\exp\lk-\frac{t^2/2}{\sigma_1^2+R_1t/3}\rk.
\end{equation*}
\end{lemma}

The second one is the vector-valued Bernstein inequality, which controls sums of independent mean-zero random vectors in a real Hilbert space.

\begin{lemma}[{\cite[Theorem~12]{gross2011recovering}}]
\label{lm:vector_bernstein}
Let $\px_1,\ldots,\px_N$ be independent mean-zero random vectors in a real Hilbert space $\mathsf{H}$ with norm $\lV\,\cdot\,\rV_{\mathsf{H}}$.  
Suppose that $\lV\px_i\rV_{\mathsf{H}}\le R_2$ almost surely for some finite constant $R_2$ and set
$\sigma^2_2\ge\sum_{i=1}^N\E\lV\px_i\rV^2_{\mathsf{H}}$.
Then, for every $0<t\le\sigma_2^2/R_2$,
\begin{equation*}
        \mathbb{P}\lk\lV\sum_{i=1}^N\px_i\rV_{\mathsf{H}}\ge t\rk
        \le
        \exp\lk-\frac{t^2}{4\sigma_2^2}+\frac{1}{4}\rk.
\end{equation*}
\end{lemma}

\subsubsection{Robust Injectivity and Dual Certificates}

We record the recovery ingredients that turn the construction of an approximate dual certificate, together with a robust injectivity event, into exact recovery.
The following proposition establishes a robust injectivity estimate on the tangent space.

\begin{proposition}[{\cite[Lemma~3.7]{candes2015cdp}},{\cite[Proposition~8]{gross2017improved}}]\label{prop:robust}
Under the random mask model in Assumption~\ref{ass:mask}, with probability at least $1-2n\exp\lk-c\frac{\nu^4L}{M^8}\rk$,
the estimate
\begin{equation}\label{eq:injectivity}
        \frac{1}{\nu^2nL}\lV\mA(\Z)\rV_2^2 
        \ge\frac{1}{4}\lV\Z\rV_F^2.
\end{equation}
holds simultaneously for all $\Z\in T$, where $c>0$ is an absolute constant.
Here, $\mathcal{A}$ is given in~\eqref{eq:sampling} and, $M$ and $\nu$ are the mask parameters specified in Assumption~\ref{ass:mask}.
\end{proposition}

\begin{remark}\label{rem:robust}
In particular, Proposition~\ref{prop:robust} provides the desired uniform lower bound on $T$ with high probability whenever $L=\mathcal{O}\lk\log n\rk$.
This estimate was first established in~\cite[Lemma~3.7]{candes2015cdp} with $L=\mathcal{O}\lk\log^3 n\rk$.
For real-valued mask ensembles,~\cite[Proposition~8]{gross2017improved} reduced the required number of masks to $L=\mathcal{O}\lk\log n\rk$.
Combining the near isotropicity identity in Proposition~\ref{prop:isotropicity} with the argument of~\cite[Proposition~8]{gross2017improved} yields Proposition~\ref{prop:robust} under the random mask model in Assumption~\ref{ass:mask}.
Since the proof follows the same lines, we omit the details.
\end{remark}

The robust injectivity estimate provides the required control on perturbations in the tangent space.
To rule out arbitrary perturbations in the feasible set of the PhaseLift program~\eqref{eq:PhaseLift}, we additionally require an approximate dual certificate.

\begin{definition}\label{def:dual}
A matrix $\Y\in\mathcal{H}_n$ is an approximate dual
certificate if
\begin{equation*}
        \Y\in \operatorname{range}\lk\mA^*\rk+\operatorname{span}\left\{\I\right\},
\end{equation*}
and
\begin{equation*}
        \lV\Y_T-\X_0\rV_F\le\frac{\nu}{4M^2\sqrt{n}},
        \qquad
        \lV\Y_{T^\perp}\rV_{\mathrm{op}}\le\frac{1}{2}.
\end{equation*}
\end{definition}

Finally, the following deterministic criterion shows that the robust injectivity estimate, together with an approximate dual certificate, guarantees uniqueness of the feasible solution.

\begin{proposition}[{\cite[Proposition~12]{gross2017improved}}]\label{prop:recovery}
Suppose that the signal energy $\lV\x_0\rV_2^2$ is known.
Assume further that~\eqref{eq:injectivity} in Proposition~\ref{prop:robust} holds and that there exists an approximate dual certificate $\Y$ in the sense of Definition~\ref{def:dual}.
Then $\X_0=\x_0\x_0^*$ is the unique feasible point of the PhaseLift feasibility program~\eqref{eq:PhaseLift}.
\end{proposition}

\subsection{Adaptive Truncated Operators}\label{subsec:adaptive}

The following adaptive truncation is designed for the iterative construction of the approximate dual certificate. 
At each stage, it discards those measurements whose interaction with the current residual is unusually large.
In the earlier CDP analyses~\cite{candes2015cdp,gross2017improved}, the cutoff level is chosen at a fixed logarithmic scale throughout the golfing construction.
Here we use a stage-adaptive truncation: for each batch, the current direction $\Z$ is fixed, and the truncation level is controlled by the parameter $\tau$.

For a fixed nonzero $\Z\in T$ and $\tau\ge1$, define the truncation events associated with an individual mask and with the $\ell$-th mask by
\begin{equation*}
\left\{
\begin{aligned}
        U_k^\tau\lk\Z\rk
        &:=
        \bigl\{\lv\tr\lk\Fm_k\Z\rk\rv \le K_0 M^2\tau\lV\Z\rV_F \bigr\},  \\
        U_{k,\ell}^{\tau}\lk\Z\rk
        &:=
        \bigl\{\lv\tr\lk\Fm_{k,\ell}\Z\rk\rv\le K_0 M^2\tau\lV\Z\rV_F\bigr\}
\end{aligned}
\right.
\end{equation*}
where $K_0>0$ is a sufficiently large numerical constant.

The adaptively truncated operator for an individual mask is defined by 
\begin{equation} 
\mathcal{F}_\tau^{\Z}\lk\W\rk
:= 
\frac{1}{\nu^2 n} \sum_{k=1}^n \ind_{U_k^\tau\lk\Z\rk} \Fm_k \tr\lk\Fm_k\W\rk. 
\end{equation}
For the $\ell$-th mask, we analogously define
\begin{equation*} 
\mathcal{F}_{\tau,\ell}^{\Z}(\W) 
:= 
\frac{1}{\nu^2 n} \sum_{k=1}^n \ind_{U_{k,\ell}^\tau\lk\Z\rk} \Fm_{k,\ell}\tr\lk\Fm_{k,\ell}\W\rk. 
\end{equation*} 
The associated truncated measurement operator formed from $L$ masks is defined by 
\begin{equation} 
\mathcal{R}_\tau^{\Z}\lk\W\rk 
:= 
\frac{1}{L}\sum_{\ell=1}^L \mathcal{F}_{\tau,\ell}^{\Z}\lk\W\rk
=\frac{1}{\nu^2 nL} \sum_{\ell=1}^L \sum_{k=1}^n \ind_{U_{k,\ell}^\tau\lk\Z\rk} \Fm_{k,\ell}\tr\lk\Fm_{k,\ell}\W\rk. 
\end{equation} 
%As $\tau\to\infty$, the truncation indicators converge to one, and hence $\mathcal{R}_\tau^{\Z}\lk\W\rk\to\mathcal{R}\lk\W\rk$.

\subsection{Bias of the Adaptive Truncation}\label{subsec:Bias}

We now estimate the bias introduced by the adaptive truncation. 
%For a fixed direction $\Z\in T$, this reduces to controlling the contribution of the discarded measurements,
%$\E\lz\lk\mathcal{F}-\mathcal{F}_\tau^{\Z}\rk\lk\Z\rk\rz$. 
The following proposition shows that, for a fixed direction $\Z\in T$, the truncation bias
$\lV\E\lz\lk\mathcal{F}-\mathcal{F}_\tau^{\Z}\rk\lk\Z\rk\rz\rV_{\mathrm{op}}$ is exponentially small in the truncation parameter $\tau$.

\begin{proposition}\label{prop:bias}
Suppose that the random masks satisfy Assumption~\ref{ass:mask} with parameters $M$ and $\nu$.
Let $\Z\in T$, and $\tau\ge1$.  
If the numerical constant $K_0>0$ in the truncation events is chosen sufficiently large,
then there exist an absolute constant $C_0>0$ such that
\begin{equation}\label{eq:EofF}
        \lV\E\lz\lk\mathcal{F}-\mathcal{F}_\tau^{\Z}\rk\lk\Z\rk\rz\rV_{\mathrm{op}}
        \le
        C_0\frac{M^4}{\nu^2}e^{-\tau}\lV\Z\rV_F.
\end{equation}
Consequently,
\begin{equation*}
\left\{
\begin{aligned}
& \lV\pmb{P}_{T^\perp}\E\lz\lk\mathcal{F}-\mathcal{F}_\tau^{\Z}\rk\lk\Z\rk\rz \rV_{\mathrm{op}}
        \le
        C_0\frac{M^4}{\nu^2}e^{-\tau}\lV\Z\rV_F,\\[6pt]
&\lV\pmb{P}_T\E\lz\lk\mathcal{F}-\mathcal{F}_\tau^{\Z}\rk\lk\Z\rk\rz\rV_F
        \le
        C_0\frac{M^4}{\nu^2}e^{-\tau}\lV\Z\rV_F.
\end{aligned}
\right.
\end{equation*}
\end{proposition}

\begin{proof}
By homogeneity, it suffices to prove the estimate under the normalization $\lV\Z\rV_F=1$.
Define $U_k^{\tau,c}\lk\Z\rk:=\{\lv\tr\lk\Fm_k\Z\rk\rv > K_0M^2\tau\}$.
For any unit vector $\xu\in\C^n$, we have\begin{equation}\label{eq:quad}
\begin{aligned}
        &\lv\xu^*\E\lz\lk\mathcal{F}-\mathcal{F}_\tau^{\Z}\rk\lk\Z\rk\rz\xu\rv\\
        &\quad\quad\quad\le
        \frac{1}{\nu^2 n}
        \sum_{k=1}^n
        \E\lz\lv\tr\lk\Fm_k\Z\rk\rv\cdot\lv\lg \D^*\vf_k,\xu\rg\rv^2\cdot\ind_{U_k^{\tau,c}\lk\Z\rk}\rz.
\end{aligned}
\end{equation}

Since $\Z\in T$, the matrix $\Z$ is Hermitian and has rank at most two. 
Hence we write
\begin{equation*}
\Z=\lambda_1\xv_1\xv_1^*+\lambda_2\xv_2\xv_2^*,
\end{equation*}
where $\lambda_1,\lambda_2\in\mathbb R$, $\lambda_1^2+\lambda_2^2=1$, and $\xv_1,\xv_2$ are orthonormal.
Consequently,
\begin{equation*}
\tr\lk\Fm_k\Z\rk
=
\lambda_1\lv\lg \vf_k,\D\xv_1\rg\rv^2+\lambda_2\lv\lg \vf_k,\D\xv_2\rg\rv^2.
\end{equation*}
By Lemma~\ref{lm:subgaussian}, we obtain
\begin{equation*} 
\mathbb{P}\lk \lv\lg \vf_k,\D\xv_r\rg\rv^2 
        > 4 M^2\tau \rk \le 4e^{-\tau},\qquad r=1,2.
\end{equation*}
Since $\lv\lambda_1\rv+\lv\lambda_2\rv\le \sqrt2$, choosing $K_0$ sufficiently large gives
\begin{equation}\label{eq:tail}
\E\lz\ind_{U_k^{\tau,c}\lk\Z\rk}\rz
=\mathbb P\lk\lv\tr\lk\Fm_k\Z\rk\rv>K_0M^2\tau\rk
\le
8e^{-\tau}.
\end{equation}
Moreover, Lemma~\ref{lm:subgaussian} with $p=8$ gives
\begin{equation}\label{eq:8}
\lk\E\lv\lg \vf_k,\D\xv_r\rg\rv^8\rk^{1/4}
\le
64M^2,\qquad r=1,2.
\end{equation}
Hence, by the triangle inequality in $L^4$,
\begin{equation}\label{eq:4}
\lk\E\lv\tr(\Fm_k\Z)\rv^4\rk^{1/4}
\le
\sum_{r=1}^2\lv\lambda_r\rv\lk\E\lv\lg \vf_k,\D\xv_r\rg\rv^8\rk^{1/4}
\le
128M^2.
\end{equation}

By Hölder's inequality, we obtain
\begin{equation*}
\begin{aligned}
&\E\lz\lv\tr\lk\Fm_k\Z\rk\rv\cdot\lv\lg \vf_k,\D\xu\rg\rv^2
\cdot\ind_{U_k^{\tau,c}\lk\Z\rk}\rz   \\                                                    
&\qquad\quad\le
\lk\E\lz\lv\tr\lk\Fm_k\Z\rk\rv^2\cdot\lv\lg \vf_k,\D\xu\rg\rv^4\rz\rk^{1/2}
\cdot\E\lz\ind_{U_k^{\tau,c}\lk\Z\rk}\rz^{1/2}\\
&\qquad\quad\le
\lk\E\lv\tr(\Fm_k\Z)\rv^4\rk^{1/4}
\cdot\lk\E\lv\lg \vf_k,\D\xu\rg\rv^8\rk^{1/4} 
\cdot\E\lz\ind_{U_k^{\tau,c}\lk\Z\rk}\rz^{1/2} \\              
&\qquad\quad\le
C_0 M^4e^{-\tau}.
\end{aligned}
\end{equation*}
In the last line, we used in order~\eqref{eq:4}, \eqref{eq:8}, and~\eqref{eq:tail}.
Substituting this estimate into~\eqref{eq:quad} gives
\begin{equation*}
\lv\xu^*\E\lz\lk\mathcal{F}-\mathcal{F}_\tau^{\Z}\rk\lk\Z\rk\rz\xu\rv
\le
C_0\frac{M^4}{\nu^2}e^{-\tau},
\qquad \lV\xu\rV_2=1.
\end{equation*}
Taking the supremum over all unit vectors $\xu$ gives~\eqref{eq:EofF}.

By the second identity in~\eqref{eq:proj0}, the map $\pmb{P}_{T^\perp}$ is a compression by the orthogonal projector $\I-\X_0$. 
Hence
\begin{equation*}
\lV\pmb{P}_{T^\perp}\W\rV_{\mathrm{op}}
\le
\lV\W\rV_{\mathrm{op}},
\end{equation*}
and the $T^\perp$ estimate follows from~\eqref{eq:EofF}.
Similarly, the first identity in~\eqref{eq:proj0} implies $\lV\pmb{P}_T\W\rV_{\mathrm{op}}\le 3\lV\W\rV_{\mathrm{op}}$. 
Moreover, $\pmb{P}_T\W\in T$ has rank at most two, and therefore
\begin{equation}\label{eq:Tbound}
\lV\pmb{P}_T\W\rV_F
\le
\sqrt2\lV\pmb{P}_T\W\rV_{\mathrm{op}}
\le
3\sqrt{2}\lV\W\rV_{\mathrm{op}}.
\end{equation}
Applying this with $\W=\E\lz\lk\mathcal{F}-\mathcal{F}_\tau^{\Z}\rk\lk\Z\rk\rz$ completes the proof.

\end{proof}

\subsection{Radius and Variance Bounds}\label{subsec:useful}

We first establish the radius and variance bounds needed for the subsequent Bernstein estimates for the adaptive truncated operators. 
The following lemma provides the required radius bounds for the individual summands.

\begin{lemma}\label{lm:R}
Suppose that the random masks satisfy Assumption~\ref{ass:mask} with parameters $M$ and $\nu$.
There exists an absolute constant $C_1>0$ such that, for every fixed $\Z\in T$, every $\tau\ge 1$,
\begin{equation*}
        \lV \mathcal{F}_{\tau}^{\Z}\lk\Z\rk \rV_{\mathrm{op}}
        \le
        C_1\frac{M^4\tau}{\nu^2}\lV\Z\rV_F.
\end{equation*}
Moreover, 
\begin{equation*} 
\left\{ 
\begin{aligned} 
&\lV\mathcal{F}_{\tau}^{\Z}\lk\Z\rk-\E \mathcal{F}_{\tau}^{\Z}\lk\Z\rk  \rV_{\mathrm{op}} 
\le 
C_1\frac{M^4\tau}{\nu^2}\lV\Z\rV_F, \\[4pt] 
&\lV \pmb{P}_T\lk\mathcal{F}_{\tau}^{\Z}\lk\Z\rk-\E \mathcal{F}_{\tau}^{\Z}\lk\Z\rk\rk \rV_F 
\le 
C_1\frac{M^4\tau}{\nu^2}\lV\Z\rV_F. 
\end{aligned} \right.
\end{equation*}
\end{lemma}

\begin{proof}
For each $k$, set $c_k:=\ind_{U_{k}^\tau\lk\Z\rk}\tr\lk\Fm_{k}\Z\rk$.
By the definition of $U_k^\tau\lk\Z\rk$, we have $\lv c_k\rv \le a:=K_0M^2\tau\lV\Z\rV_F$.
The orthogonality relations for $\left\{\vf_k\right\}_{k=1}^n$, together with Assumption~\ref{ass:mask} gives
\begin{equation*} 
        \lV \sum_{k=1}^n\Fm_{k}\rV_{\mathrm{op}}=\lV\D^*\lk\sum_{k=1}^n \vf_k\vf_k^*\rk\D\rV_{\mathrm{op}}=n\lV\D^*\D\rV_{\mathrm{op}}\le nM^2.
\end{equation*}
Since $\Fm_{k} \succeq\pmb{0}$ and $\lv c_k\rv\le a$ for every $k$, we have
\begin{equation*} 
        -a\sum_{k=1}^n\Fm_{k}
        \preceq
        \sum_{k=1}^n c_k\Fm_{k}
        \preceq
        a\sum_{k=1}^n\Fm_{k}.
\end{equation*}
It follows that
\begin{equation*}
\begin{aligned}
        \lV\mathcal{F}_{\tau}^{\Z}\lk\Z\rk\rV_{\mathrm{op}}
        =
        \frac{1}{\nu^2 n}\lV\sum_{k=1}^nc_k\Fm_{k}\rV_{\mathrm{op}} 
        \le
        \frac{a}{\nu^2 n}\lV\sum_{k=1}^n\Fm_{k}\rV_{\mathrm{op}}
        \le C_2\frac{M^4\tau}{\nu^2}\lV\Z\rV_F.
\end{aligned}
\end{equation*}
By Jensen's inequality, the same bound holds for $\lV\E \mathcal{F}_{\tau}^{\Z}\lk\Z\rk\rV_{\mathrm{op}}$.
Hence, the triangle inequality then gives
\begin{equation*} 
        \lV \mathcal{F}_{\tau}^{\Z}\lk\Z\rk-\E \mathcal{F}_{\tau}^{\Z}\lk\Z\rk\rV_{\mathrm{op}}
        \le
        \lV \mathcal{F}_{\tau}^{\Z}\lk\Z\rk\rV_{\mathrm{op}}+\lV\E \mathcal{F}_{\tau}^{\Z}\lk\Z\rk\rV_{\mathrm{op}}
        \le
        2C_2\frac{M^4\tau}{\nu^2}\lV\Z\rV_F .
\end{equation*}
Finally, applying ~\eqref{eq:Tbound} to $\W=\mathcal{F}_{\tau}^{\Z}\lk\Z\rk-\E \mathcal{F}_{\tau}^{\Z}\lk\Z\rk$ gives the last estimate.
Enlarging the absolute constant if necessary and denoting it by $C_1$ proves all three estimates.
\end{proof}

The next two lemmas provide the corresponding variance bounds for the individual summands.

\begin{lemma}\label{lm:variance1}
Suppose that the random masks satisfy Assumption~\ref{ass:mask} with parameters
$M$ and $\nu$.
There exists an absolute constant $C_3>0$ such that, for every fixed $\Z\in T$, every $\tau\ge1$
\begin{equation}\label{eq:variance1}
\lV \E\lz\lk\mathcal{F}_{\tau}^{\Z}\lk\Z\rk-\E\mathcal{F}_{\tau}^{\Z}\lk\Z\rk\rk^2 \rz\rV_{\mathrm{op}}
        \le
        C_3\frac{M^8}{\nu^4}\lV\Z\rV_F^2.
\end{equation}
\end{lemma}

\begin{proof}
Since $\mathcal{F}_{\tau}^{\Z}\lk\Z\rk$ is Hermitian, the variance identity gives, 
\begin{equation*} 
\E \lz\lk\mathcal{F}_{\tau}^{\Z}\lk\Z\rk - \E \mathcal{F}_{\tau}^{\Z}\lk\Z\rk\rk^2 \rz
= 
\E \lz\mathcal{F}_{\tau}^{\Z}\lk\Z\rk^2\rz-\lk\E \lz\mathcal{F}_{\tau}^{\Z}\lk\Z\rk\rz\rk^2 
\preceq \E \lz\mathcal{F}_{\tau}^{\Z}\lk\Z\rk^2\rz. 
\end{equation*} 
It therefore suffices to prove that
\begin{equation}\label{eq:single}
       \lV\E \lz\mathcal{F}_{\tau}^{\Z}\lk\Z\rk^2\rz\rV_{\mathrm{op}}
       \le
        C_3\frac{M^8}{\nu^4 }\lV\Z\rV_F^2.
\end{equation}

Fix a unit vector $\xu\in\C^n$.  
Then $\xu^* \E \lz\mathcal{F}_{\tau}^{\Z}\lk\Z\rk^2\rz\xu=\E \lV \mathcal{F}_{\tau}^{\Z}\lk\Z\rk \xu\rV_2^2$.
Moreover,
\begin{equation*}
        \mathcal{F}_{\tau}^{\Z}\lk\Z\rk \xu
        =
        \frac{1}{\nu^2n}\sum_{k=1}^n\ind_{U_{k}^{\tau}\lk\Z\rk}\tr\lk\Fm_{k}\Z\rk\,\D^*\vf_k\lg\D^*\vf_k,\xu\rg.
\end{equation*}
For arbitrary coefficients $\{\widetilde{c}_k\}_{k=1}^n$, Assumption~\ref{ass:mask} and Fourier orthogonality give
\begin{equation*}
        \lV\sum_{k=1}^n \widetilde{c}_k\D^*\vf_k\rV_2^2
        \le M^2\lV\sum_{k=1}^n \widetilde{c}_k \vf_k\rV_2^2
        = M^2 n \sum_{k=1}^n\lv \widetilde{c}_k\rv^2.
\end{equation*}
Taking
\begin{equation*}
        \widetilde{c}_k=
        \ind_{U_{k}^{\tau}\lk\Z\rk}
        \tr\lk\Fm_{k}\Z\rk\lg\D^*\vf_k,\xu\rg,
\end{equation*}
it follows that
\begin{equation}\label{eq:Sl}
\begin{aligned}
\E \lV \mathcal{F}_{\tau}^{\Z}\lk\Z\rk  \xu\rV_2^2
        &\le
        \frac{M^2}{\nu^4 n }
        \sum_{k=1}^n
        \E\lz\ind_{U_{k}^{\tau}\lk\Z\rk}\cdot\lv\tr\lk\Fm_{k}\Z\rk\rv^2\cdot\lv\lg\D^*\vf_k,\xu\rg\rv^2\rz \\
        &\le
        \frac{M^2}{\nu^4 n }\sum_{k=1}^n 
        \E\lz\lv\tr\lk\Fm_{k}\Z\rk\rv^2\cdot\lv\lg\D^*\vf_k,\xu\rg\rv^2\rz\\
        &\le
        \frac{M^2}{\nu^4n }\sum_{k=1}^n 
        \lk\E\lv\tr\lk\Fm_{k}\Z\rk\rv^4\rk^{1/2}\cdot\lk\E\lv\lg\D^*\vf_k,\xu\rg\rv^4\rk^{1/2}\\
        &\le C_3 \frac{M^8}{\nu^4 }\lV\Z\rV_F^2.
\end{aligned}
\end{equation}
Here, the second inequality follows from $\ind_{U_k^\tau\lk\Z\rk}\le1$, the third from the Cauchy--Schwarz inequality, and the last from the fourth moment bound~\eqref{eq:4} and
Lemma~\ref{lm:subgaussian} with $p=4$, applied to the unit vector $\xu$.

Since $\E \lz\mathcal{F}_{\tau}^{\Z}\lk\Z\rk^2\rz\succeq\pmb{0}$, taking the supremum over all $\lV\xu\rV_2=1$ proves~\eqref{eq:single}, and hence the desired result.
\end{proof}

\begin{lemma}\label{lm:variance2}
Suppose that the random masks satisfy Assumption~\ref{ass:mask} with parameters $M$ and $\nu$. 
There exists an absolute constant $C_4>0$ such that, for every fixed $\Z\in T$ and every $\tau\ge1$,
\begin{equation*}
       \E\lz\lV\pmb{P}_T\lk \mathcal{F}_{\tau}^{\Z}\lk\Z\rk-\E \mathcal{F}_{\tau}^{\Z}\lk\Z\rk\rk\rV_F^2\rz
        \le
    C_4\frac{M^8}{\nu^4}\lV\Z\rV_F^2.
\end{equation*}
\end{lemma}

\begin{proof}
Since $\pmb{P}_T$ is linear, the variance identity gives
\begin{equation*} 
\begin{aligned} 
\E\lV\pmb{P}_T\lk \mathcal{F}_{\tau}^{\Z}\lk\Z\rk-\E \mathcal{F}_{\tau}^{\Z}\lk\Z\rk\rk\rV_F^2 
&= 
\E\lV\pmb{P}_T \mathcal{F}_{\tau}^{\Z}\lk\Z\rk\rV_F^2-\lV\E\pmb{P}_T \mathcal{F}_{\tau}^{\Z}\lk\Z\rk\rV_F^2\\
&\le 
\E\lV\pmb{P}_T \mathcal{F}_{\tau}^{\Z}\lk\Z\rk\rV_F^2. 
\end{aligned} 
\end{equation*}
It therefore suffices to bound $\E\lV\pmb{P}_T \mathcal{F}_{\tau}^{\Z}\lk\Z\rk\rV_F^2$. 

For a Hermitian matrix $\W$, set $\y:=\W\x_0$.
By the projection formula~\eqref{eq:proj0},
\begin{equation*}
\pmb{P}_T\W = \x_0\y^*+\y\x_0^* - \lk\x_0^*\W\x_0\rk \x_0\x_0^*.
\end{equation*}
Using $\lV\x_0\rV_2=1$, we obtain
\begin{equation} 
\begin{aligned} 
\lV\pmb{P}_T\W\rV_F 
&\le \lV\x_0\y^*\rV_F + \lV\y\x_0^*\rV_F + \lv\x_0^*\W\x_0\rv\cdot\lV\x_0\x_0^*\rV_F\\ 
&\le 2\lV\y\rV_2 + \lv \lg \x_0,\y\rg\rv \le 3\lV\W\x_0\rV_2. 
\end{aligned} 
\end{equation}
Hence, 
\begin{equation*} 
\E\lV\pmb{P}_T \mathcal{F}_{\tau}^{\Z}\lk\Z\rk\rV_F^2 
\le 
9\E\lV \mathcal{F}_{\tau}^{\Z}\lk\Z\rk\x_0\rV_2^2
\le
C_4\frac{M^8}{\nu^4}\lV\Z\rV_F^2,
\end{equation*} 
where the second inequality follows from the vector estimate~\eqref{eq:Sl} from Lemma~\ref{lm:variance1}, applied with $\xu=\x_0$.
Combining the preceding estimates proves the result.
\end{proof}

\subsection{\texorpdfstring{$T^\perp$}{Tperp} and \texorpdfstring{$T$}{T} Estimates}\label{subsec:two_estimates}

We now combine the bias estimate from Section~\ref{subsec:Bias} with the radius and variance bounds from Section~\ref{subsec:useful} to obtain high probability estimates for the $T^\perp$ and $T$ components of the adaptive truncated operators. 

We begin with the $T^\perp$ component.

\begin{proposition}\label{prop:tperp}
Suppose that the random masks satisfy Assumption~\ref{ass:mask} with parameters $M$ and $\nu$, and let $C_0>0$ be the constant from Proposition~\ref{prop:bias}.
There exists an absolute constant $\widetilde{c}_1>0$ such that the following holds.
For every fixed $\Z\in T$, every $\tau\ge1$, and every $t_1>0$ satisfying
\begin{equation}\label{eq:asmp1}
C_0\frac{M^4}{\nu^2}e^{-\tau}
\le\frac{t_1}{4},
\end{equation}
we have
\begin{equation}
\begin{aligned}
        &\mathbb{P}\lk\lV\pmb{P}_{T^\perp} \lk\mathcal{R}_\tau^{\Z}\lk\Z\rk-\tr\lk\Z\rk\I\rk\rV_{\mathrm{op}}
        \ge t_1\lV\Z\rV_F\rk\\
        &\qquad\qquad\qquad\qquad\qquad\le 2n\exp\lk-\widetilde{c}_1\frac{\nu^4L}{M^8\tau}\min\left\{t_1^2,t_1\right\}\rk.
        \end{aligned}
\end{equation}
\end{proposition}

\begin{proof} 
By homogeneity, we may assume without loss of generality $\lV\Z\rV_F=1$.
By Proposition~\ref{prop:isotropicity}, $\E\mathcal{R}\lk\Z\rk=\Z+\tr\lk\Z\rk\I$.
Since $\Z\in T$, we have $\pmb{P}_{T^\perp}\Z=0$, and hence
\begin{equation*}
        \pmb{P}_{T^\perp}\E\mathcal{R}\lk\Z\rk
        =
        \tr\lk\Z\rk\pmb{P}_{T^\perp}\I.
\end{equation*}
Therefore, by the triangle inequality and the projection formula~\eqref{eq:proj0}, we have
\begin{equation}\label{eq:Perp}
\begin{aligned}
        &\lV\pmb{P}_{T^\perp}\lk\mathcal{R}_\tau^{\Z}\lk\Z\rk-\tr\lk\Z\rk\I\rk\rV_{\mathrm{op}}\\
        &\qquad
        \le
        \lV\pmb{P}_{T^\perp}\lk \mathcal{R}_\tau^{\Z}\lk\Z\rk-\E\mathcal R_\tau^{\Z}\lk\Z\rk\rk\rV_{\mathrm{op}}
        +
        \lV\pmb{P}_{T^\perp}\lk \E\lk\mathcal{R}_\tau^{\Z}-\mathcal R\rk\lk\Z\rk\rk\rV_{\mathrm{op}}\\
        &\qquad
        \le
        \lV\mathcal{R}_\tau^{\Z}\lk\Z\rk-\E\mathcal R_\tau^{\Z}\lk\Z\rk\rV_{\mathrm{op}}
        +
        \lV \E\lk\mathcal{R}_\tau^{\Z}-\mathcal R\rk\lk\Z\rk\rV_{\mathrm{op}}.
\end{aligned}
\end{equation}
By Proposition~\ref{prop:bias} and the assumption~\eqref{eq:asmp1}, the second term on the right hand side is bounded by $t_1/4$.

Let
\begin{equation*}
        \pmb{S}_\ell
        :=\frac{1}{L}\lk\mathcal{F}_{\tau,\ell}^{\Z}\lk\Z\rk -\E \mathcal{F}_{\tau,\ell}^{\Z}\lk\Z\rk\rk.
\end{equation*}
Then $\pmb{S}_1,\ldots,\pmb{S}_L$ are independent centered Hermitian random matrices, and
\begin{equation*}
        \mathcal{R}_\tau^{\Z}\lk\Z\rk-\E\mathcal{R}_\tau^{\Z}\lk\Z\rk
        =
        \sum_{\ell=1}^L\pmb{S}_\ell.
\end{equation*}
By Lemma~\ref{lm:R} and Lemma~\ref{lm:variance1}, the centered summands satisfy
\begin{equation*}
\left\{
\begin{aligned}
&\lV\pmb{S}_\ell\rV_{\mathrm{op}}
        \le \widetilde{C}_1\frac{M^4\tau}{\nu^2 L}
        :=R_1,\\[6pt]
&\lV \sum_{\ell=1}^L
        \E\lz\pmb{S}_\ell^2\rz
        \rV_{\mathrm{op}}
        \le
        \widetilde{C}_2\frac{M^8}{\nu^4 L}
        :=\sigma_1^2.
\end{aligned}
\right.
\end{equation*}
Applying the matrix Bernstein inequality in Lemma~\ref{lm:matrix_bernstein} with $3t_1/4$, we obtain
\begin{equation}\label{eq:bernstein1}
\begin{aligned}
        \mathbb{P}\lk \lV\mathcal{R}_\tau^{\Z}\lk\Z\rk-\E\mathcal{R}_\tau^{\Z}\lk\Z\rk\rV_{\mathrm{op}}
        \ge
        \frac{3}{4} t_1 \rk
        &\le
        2n\exp\lk-c_1\frac{t_1^2}{\sigma_1^2+R_1t_1}\rk\\
        &\le
        2n\exp\lk-\widetilde{c}_1\frac{\nu^4L}{M^8\tau}\min\{t_1^2,t_1\}\rk.
\end{aligned}
\end{equation}
In the last step, we used $\tau\ge1$ and the parameter relations in Assumption~\ref{ass:mask}.

Combining~\eqref{eq:Perp} with~\eqref{eq:bernstein1} gives the desired bound.
\end{proof}

The second proposition gives the corresponding $T$ estimate.

\begin{proposition}\label{prop:t}
Suppose that the random masks satisfy Assumption~\ref{ass:mask} with parameters $M$ and $\nu$, and let $C_0>0$ be the constant from Proposition~\ref{prop:bias}. 
Then there exists an absolute constant $\widetilde{c}_2>0$ such that the following holds.
For every fixed $\Z\in T$, every $\tau\ge1$, and every $0<t_2\le1$ satisfying
\begin{equation}\label{eq:asmp2} 
C_0\frac{M^4}{\nu^2}e^{-\tau} \le \frac{t_2}{4},
\end{equation} 
we have
\begin{equation}\label{eq:t}
\begin{aligned}
\mathbb{P}\lk\lV\pmb{P}_T \lk\mathcal{R}_\tau^{\Z}\lk\Z\rk-\Z-\tr\lk\Z\rk\I\rk\rV_F 
\ge t_2\lV\Z\rV_F\rk
\le \exp\lk-\widetilde{c}_2 \frac{t_2^2\nu^4L}{M^8\tau} +\frac{1}{4}\rk. 
\end{aligned}
\end{equation}
\end{proposition}

\begin{proof}  
By homogeneity, it suffices to consider $\lV\Z\rV_F=1$.
By Proposition~\ref{prop:isotropicity}, we obtain
\begin{equation}\label{eq:tt}
\begin{aligned}
        &\lV\pmb{P}_T\lk\mathcal{R}_\tau^{\Z}\lk\Z\rk-\Z-\tr\lk\Z\rk\I\rk\rV_F\\
        &\qquad\qquad
        \le
        \lV\pmb{P}_T\lk\mathcal{R}_\tau^{\Z}\lk\Z\rk-\E\mathcal{R}_\tau^{\Z}\lk\Z\rk \rk\rV_F
        +
        \lV\pmb{P}_T \E\lk\mathcal{R}_\tau^{\Z}-\mathcal R\rk\lk\Z\rk\rV_F.
\end{aligned}
\end{equation}
By Proposition~\ref{prop:bias} and the assumed condition~\eqref{eq:asmp2}, the second term on the right hand side is at most $t_2/4$.

Let
\begin{equation*}
        \px_\ell
        :=\frac{1}{L}\pmb{P}_T\lk\mathcal{F}_{\tau,\ell}^{\Z}\lk\Z\rk -\E \mathcal{F}_{\tau,\ell}^{\Z}\lk\Z\rk\rk.
\end{equation*}
Then $\px_1,\ldots,\px_L$ are independent mean-zero random vectors in the real Hilbert space $T$, and
\begin{equation*}
\pmb{P}_T \lk \mathcal{R}_\tau^{\Z}\lk\Z\rk-\E\mathcal{R}_\tau^{\Z}\lk\Z\rk\rk
= \sum_{\ell=1}^L\px_\ell. 
\end{equation*}
By Lemma~\ref{lm:R}, Lemma~\ref{lm:variance2} and $\tau\ge1$, we have 
\begin{equation*} 
\left\{ \begin{aligned}
&\lV\px_\ell\rV_F \le 
\widetilde{C}_3\frac{M^4\tau}{\nu^2L}
:=R_2, \\[6pt] 
&\sum_{\ell=1}^L \E\lV\px_\ell\rV_F^2 
\le \widetilde{C}_4\frac{M^8}{\nu^4L}
\le \widetilde{C}_4\frac{M^8\tau}{\nu^4L}
:=\sigma_2^2. 
\end{aligned} \right. 
\end{equation*} 
After enlarging $\widetilde{C}_4$ if necessary, the parameter relations in Assumption~\ref{ass:mask} imply
\begin{equation*}
\frac{\sigma_2^2}{R_2}
=
\frac{\widetilde{C}_4}{\widetilde{C}_3} \frac{M^4}{\nu^2}
\ge1.
\end{equation*}
Since $0<t_2\le1$, the threshold $\frac{3}{4}t_2\le\sigma_2^2/R_2$.
Thus, the vector-valued Bernstein inequality in Lemma~\ref{lm:vector_bernstein} gives
\begin{equation}\label{eq:bernstein2} 
\mathbb{P}\lk \lV \pmb{P}_T \lk \mathcal{R}_\tau^{\Z}\lk\Z\rk-\E\mathcal{R}_\tau^{\Z}\lk\Z\rk\rk \rV_F \ge \frac{3}{4} t_2 \rk 
\le 
\exp\lk -\widetilde{c}_2\frac{t_2^2\nu^4L}{M^8\tau} +\frac{1}{4} \rk. 
\end{equation}

Combining~\eqref{eq:bernstein2} with~\eqref{eq:tt} gives the desired bound.
\end{proof}

\subsection{Golfing Construction}\label{subsec:golfing}

The approximate dual certificate is obtained by an improved golfing construction based on~\cite{candes2015cdp,gross2017improved}.
The first modification concerns the allocation of masks.
Although the construction still uses $\mathcal{O}\lk\log n\rk$ successful golfing updates, the number of masks assigned to the updates decreases as the construction progresses.
This differs from the fixed batch size strategy in~\cite{candes2015cdp} and from the two stage scheme in~\cite{gross2017improved}.
The second modification is that each golfing update is built from an adaptive truncated operator with a dimension-independent truncation threshold (up to the mask distribution parameters), rather than from the logarithmic scale truncated operators used in~\cite{candes2015cdp,gross2017improved}.
This allows us to apply the estimates in Section~\ref{subsec:two_estimates} to the current residual at each step. 

The next proposition builds an approximate dual certificate using only $\mathcal{O}\lk\log n\rk$ masks.

\begin{proposition}\label{prop:golfing}
Suppose that the random masks satisfy Assumption~\ref{ass:mask} with parameters $M$ and $\nu$.
For any $\omega\ge1$, there exists a constant $C=C\lk M,\nu\rk>0$ such that, if the total number of coded diffraction
patterns satisfies
\begin{equation*}
        L \ge C\,\omega\,\log n,
\end{equation*}
then with probability at least $1-\frac{1}{2}n^{-\omega}$, an approximate dual certificate satisfying Definition~\ref{def:dual} can be constructed.
\end{proposition}

\begin{proof}

We partition a subset of the $L$ masks into independent batches and use only these batches in the construction.
Since the unused masks are assigned zero coefficients, the resulting certificate still belongs to $\operatorname{range}\lk\mA^*\rk+\operatorname{span}\left\{\I\right\}$ for the full measurement operator.
In the following, the proof proceeds in five steps.

\textbf{Step 1. Parameter Selection.}
The golfing scheme constructs the certificate iteratively. 
Set
\begin{equation}\label{eq:r}
        r:=
        \left\lceil\frac{1}{2}\log_2 n\right\rceil
        +\left\lceil\log_2\lk\frac{M^2}{\nu}\rk\right\rceil+2.
\end{equation}
to be the number of successful golfing updates required to satisfy Definition~\ref{def:dual}.
Fix $\eta=1/7$, so that $\eta\sum_{j=0}^{\infty}2^{-j/2}\le\frac{1}{2}$.
For $0\le j\le r-1$, define
\begin{equation}\label{eq:t1}
        t_{1,j}:=\eta 2^{j/2},\quad t_{2}:=\frac{1}{2}
        \qquad
        a_j:=\min\left\{t_{1,j}^2,t_{1,j}\right\}.
\end{equation}
Here, $t_{1,j}$ and $t_{2}$ are the admissible $T^\perp$, $T$ errors at level $j$.
Choose a constant order truncation level $\tau_0=C_1 \log\lk M^2/\nu\rk\ge1$ such that
\begin{equation}\label{eq:tau0}
        C_0\frac{M^4}{\nu^2}e^{-\tau_0}
        \le
        \frac{1}{4}\min\left\{\eta,\frac{1}{2}\right\}\le\frac{1}{28}.
\end{equation}
Here $C_0>0$ is the constant from Proposition~\ref{prop:bias} and $C_1$ is a sufficiently large numerical constant.

For $0\le j\le r-1$, let
\begin{equation}\label{eq:mj}
        m_j:=
        \left\lceil
        C_2\frac{M^8\tau_0}{\nu^4}\lk1+\frac{\log n}{a_j}\rk
        \right\rceil,
\end{equation}
be the number of masks used in each trial at level $j$, where $C_2>0$ is a
sufficiently large numerical constant to be chosen below.
Finally, set
\begin{equation*}
B:=\sum_{j=0}^{r-1}m_j
\end{equation*}
to be the baseline mask cost corresponding to one successful trial at each level.
By~\eqref{eq:t1},
\begin{equation*}
\frac1{a_j}
=
\frac1{\min\left\{\eta^2 2^j,\eta 2^{j/2}\right\}}
\le
\frac{1}{\eta^2 2^j}+ \frac{1}{\eta 2^{j/2}}.
\end{equation*}
Therefore
\begin{equation*}
        \sum_{j=0}^{r-1}\frac1{a_j}
        \le
        \sum_{j=0}^{\infty}\lk\frac{1}{\eta^2 2^j}+\frac{1}{\eta 2^{j/2}}\rk
        \le C_\eta,
\end{equation*}
By~\eqref{eq:r}, $r\lesssim_{M,\nu}\log n$ for $n\ge2$, the definition~\eqref{eq:mj} implies
\begin{equation}\label{eq:B_bound}
        B=\sum_{j=0}^{r-1}m_j\le C_3\lk M,\nu\rk \log n.
\end{equation}

\textbf{Step 2. Adaptive Golfing Construction.}
We now describe the construction. 
At each step, $\Y$ denotes the current partial certificate and
\begin{equation*}
        \Q=\X_0-\pmb{P}_T\Y
\end{equation*}
is the remaining tangent space residual.
Initially $\Y=\pmb{0}$ and $\Q=\X_0$. 
We shall construct at most $r$ successful updates.
Suppose that $j$ successful updates have already been obtained, where $0\le j\le r-1$. 
If $\Q_j=\pmb{0}$, we stop the construction.
Otherwise, we start independent trials for the $\lk j+1\rk$-th successful update, which we index as level $j$. Each trial uses a fresh batch of $m_j$ masks.

Conditionally on all previously used batches, the current residual $\Q_j\in T$ is fixed and is independent of the fresh batch.
Let $\mathcal{R}_{\tau_0,k}^{\Q_j}$ denote the corresponding $\Q_j$-adapted truncated operator formed from the $k$-th trial batch at level $j$.
For this trial, define the candidate certificate
\begin{equation}\label{eq:Y_new}
        \Y_{\mathrm{new}}
        :=
        \Y_j+ \mathcal{R}_{\tau_0,k}^{\Q_j}\lk\Q_j\rk - \tr\lk\Q_j\rk\I
\end{equation}
and the corresponding candidate residual
\begin{equation}\label{eq:Q_new}
        \Q_{\mathrm{new}}
        :=
        \X_0-\pmb{P}_T\Y_{\mathrm{new}}.
\end{equation}
Equivalently,
\begin{equation*}
        \Q_{\mathrm{new}}
        =
        \pmb{P}_T \lk \Q_j+\tr\lk\Q_j\rk\I-\mathcal{R}_{\tau_0,k}^{\Q_j}\lk\Q_j\rk\rk.
\end{equation*}

The trial is declared successful if
\begin{equation}\label{eq:success}
\left\{
\begin{aligned}
&\lV\pmb{P}_{T^\perp}
\lk\mathcal{R}_{\tau_0,k}^{\Q_j}\lk\Q_j\rk-\tr\lk\Q_j\rk\I\rk\rV_{\mathrm{op}}
\le
t_{1,j}\lV\Q_j\rV_F,\quad t_{1,j}=\eta 2^{j/2}\\[6pt]
&\lV\Q_{\mathrm{new}}\rV_F
\le
t_{2}\lV\Q_j\rV_F,\quad t_{2}=\frac{1}{2}.
\end{aligned}
\right.
\end{equation}
If the trial succeeds, we set
\begin{equation*}
        \lk\Y_{j+1},\Q_{j+1}\rk:=\lk\Y_{\mathrm{new}},\Q_{\mathrm{new}}\rk,
\end{equation*}
and proceed to level $j+1$.
If the trial fails, the pair $\lk\Y_j,\Q_j\rk$ is left unchanged, and we repeat
level $j$ with a new independent batch of masks.

\textbf{Step 3. Success Probability of One Trial.}
We next estimate the conditional success probability of one trial at level $j$.
Let $\mathcal{F}_{\mathrm{past}}$ denote the randomness generated by all previously used batches. 
Conditionally on $\mathcal{F}_{\mathrm{past}}$, the residual $\Q_j\in T$ is fixed and is independent of the fresh batch used in the current trial.
Since $t_{1,j}\ge\eta$ and $t_2=1/2$, the choice~\eqref{eq:tau0} implies the bias assumptions in both Proposition~\ref{prop:tperp} and Proposition~\ref{prop:t}.
Applying Proposition~\ref{prop:tperp} conditionally on $\mathcal{F}_{\mathrm{past}}$ gives
\begin{equation}\label{eq:tperp_fail}
\begin{aligned}
&\mathbb{P}\lk
\lV\pmb{P}_{T^\perp}\lk\mathcal{R}_{\tau_0,k}^{\Q_j}\lk\Q_j\rk-\tr\lk\Q_j\rk\I\rk\rV_{\mathrm{op}}>t_{1,j}\lV\Q_j\rV_F\,\middle|\,\mathcal{F}_{\mathrm{past}}\rk\\
%\qquad\qquad\qquad\qquad\qquad\qquad\qquad
&\qquad\qquad\qquad\qquad\qquad\qquad\qquad\qquad
\le
2n\exp\lk-c\frac{\nu^4 m_j}{M^8\tau_0}a_j\rk.
\end{aligned}
\end{equation}
Similarly, applying Proposition~\ref{prop:t} conditionally on $\mathcal{F}_{\mathrm{past}}$ with $t_2=1/2$ gives
\begin{equation}\label{eq:t_fail}
\begin{aligned}
&\mathbb{P}\lk\lV\pmb{P}_T\lk\mathcal{R}_{\tau_0,k}^{\Q_j}\lk\Q_j\rk-\Q_j-\tr\lk\Q_j\rk\I\rk\rV_F
>
\frac{1}{2}\lV\Q_j\rV_F\,\middle|\,\mathcal{F}_{\mathrm{past}}\rk \\
&\qquad\qquad\qquad\qquad\qquad\qquad\qquad\qquad\le
\exp\lk-c\frac{\nu^4 m_j}{M^8\tau_0}+\frac{1}{4}\rk.
\end{aligned}
\end{equation}
By~\eqref{eq:mj}, 
\begin{equation*}
        \frac{\nu^4 m_j}{M^8\tau_0}a_j
        \ge
        C_2\lk a_j+\log n\rk\ge C_2\log n,\qquad \frac{\nu^4 m_j}{M^8\tau_0}\ge C_2.
\end{equation*}
Thus, combining~\eqref{eq:tperp_fail} and~\eqref{eq:t_fail} and after increasing $C_2$ if necessary, there exists an absolute constant $\alpha_0$, such that every trial at level $j$ fails with conditional probability at most
\begin{equation}\label{eq:conditional_success}
        \mathbb{P}\lk\text{the current trial fails}\,\middle|\,\mathcal{F}_{\mathrm{past}}\rk
        \le
        \exp\lk-\alpha_0\frac{\nu^4 m_j}{M^8\tau_0}\rk:=p_j.
\end{equation}
%uniformly over all admissible conditioning events and all $0\le j\le r-1$.

\textbf{Step 4. Total Mask Consumption.}
Let $N_j$ be the number of trials required to obtain one successful update at level $j$.
By~\eqref{eq:conditional_success}, for every integer $l\ge0$,
\begin{equation}
        \mathbb{P}\lk N_j>l \,\middle|\, N_0,\ldots,N_{j-1}\rk\le\lk p_j\rk^l.
\end{equation}
Thus $N_j$ is conditionally stochastically dominated by a geometric random variable with success probability at least $1-p_j$.
%Notice that the variables $\left\{ N_j\right\}_{j=0}^{r-1}$ are not independent.
The total number of masks consumed before obtaining all required successful updates is bounded by
\begin{equation*}
        L_{\mathrm{total}}:=\sum_{j=0}^{r-1}m_j N_j.
\end{equation*}
We use  Lemma~\ref{lm:geometric} in Appendix~\ref{ap:geo} to bound the tail probability of $L_{\mathrm{total}}$.
Set $\theta :=\frac{\alpha_0\nu^4}{2M^8\tau_0}$ in Appendix~\ref{ap:geo}.
By~\eqref{eq:conditional_success},
we have
\begin{equation*}
        p_je^{\theta m_j}=e^{-\theta m_j}<1,
        \qquad 0\le j\le r-1.
\end{equation*}
For every $j$, it then follows that
\begin{align*}
\frac{\lk 1-p_j\rk e^{\theta m_j}}{1-p_je^{\theta m_j}}
=\frac{\lk 1-e^{-2\theta m_j}\rk e^{\theta m_j}}{1-e^{-\theta m_j}}
=1+e^{\theta m_j}\le e^{2\theta m_j}.
\end{align*}
Applying Lemma~\ref{lm:geometric} with $t=4\omega B$, we obtain
\begin{equation}\label{eq:total0}
\begin{aligned}
\mathbb{P}\lk L_{\mathrm{total}}>4\omega B\rk
&\le
e^{-4\omega\theta B}
\prod_{j=0}^{r-1}
\frac{\lk 1-p_j\rk e^{\theta m_j}}{1-p_je^{\theta m_j}}\\
&\le
\exp\lk-\lk 4\omega-2\rk\theta B\rk\\
%&\le
%\exp\lk -2\omega\theta B\rk\\
&\le
\exp\lk -\alpha_0\omega \frac{\nu^4B}{M^8\tau_0}\rk,
\end{aligned}
\end{equation}
where we used $\omega\ge1$ in the last inequality. 
By the definition of $m_j$ in~\eqref{eq:mj} and $r$ in~\eqref{eq:r},
\begin{equation*}
        \alpha_0\frac{\nu^4B}{M^8\tau_0}
        =\alpha_0\frac{\nu^4}{M^8\tau_0}\sum_{j=0}^{r-1}m_j
        \ge \alpha_0 C_2r 
        \ge \log\lk 2n\rk,
\end{equation*}
after increasing $C_2$ once more.
Consequently,~\eqref{eq:total0} yields
\begin{align*}
\mathbb{P}\lk L_{\mathrm{total}}>4\omega B\rk
&\le\exp\lk-\omega\log(2n)\rk\\
&=\lk 2n\rk^{-\omega}
\le\frac{1}{2}n^{-\omega}.
\end{align*}
By~\eqref{eq:B_bound}, $4\omega B \le C_3\lk M,\nu\rk\,\omega\log n$.
Therefore, after increasing the constant in the assumption $L\ge C(M,\nu)\,\omega\log n$ if necessary, we have $4\omega B\le L$.
Consequently, with probability at least $1-\frac{1}{2}n^{-\omega}$, the construction obtains all the required successful updates before exhausting the available $L$ masks.

\textbf{Step 5: Dual Certificate Verification.}
It remains to verify that, on the event constructed above, the resulting matrix $\Y$ is an approximate dual certificate.
By the second condition in~\eqref{eq:success}, each successful update satisfies
\begin{equation*}
        \lV\Q_{j+1}\rV_F
        \le
        \frac{1}{2}\lV\Q_j\rV_F .
\end{equation*}
Since $\lV\Q_0\rV_F=\lV\X_0\rV_F=1$ and the definition~\eqref{eq:r} of $r$, we have
\begin{equation}\label{eq:T_estimate}
        \lV\Y_T-\X_0\rV_F
        =\lV\Q_r\rV_F\le2^{-r}\le
        \frac{\nu}{4M^2\sqrt n}.
\end{equation}

We now control the $T^\perp$ component. Let $\mathcal S_j$ denote the truncated
operator associated with the batch used in the $(j+1)$-th successful update, so that
$\mathcal S_j$ is applied to $\Q_j$. Let
\begin{equation*}
        \Delta_j
        :=
        \pmb{P}_{T^\perp}\lk\mathcal{S}_j\lk\Q_j\rk-\tr\lk\Q_j\rk\I\rk,
        \qquad 0\le j\le r-1.
\end{equation*}
By the update formula,
\begin{equation*}
        \Y_{T^\perp}
        =
        \sum_{j=0}^{r-1}\Delta_j .
\end{equation*}
On the success event, the first condition in~\eqref{eq:success} gives
\begin{equation*}
        \lV\Delta_j\rV_{\mathrm{op}}
        \le
        t_{1,j}\lV\Q_j\rV_F
        \le
        \eta 2^{j/2}2^{-j}=\eta 2^{-j/2}.
\end{equation*}
Therefore,
\begin{equation}\label{eq:tperp_estimate}
        \lV\Y_{T^\perp}\rV_{\mathrm{op}}
        \le
        \eta\sum_{j=0}^{\infty}2^{-j/2}
        \le
        \frac{1}{2}.
\end{equation}

It is clear from the construction that 
$\Y\in\operatorname{range}\lk\mA^*\rk+\operatorname{span}\left\{\I\right\}$.
Combining this with~\eqref{eq:T_estimate} and~\eqref{eq:tperp_estimate}, we conclude that $\Y$ satisfies Definition~\ref{def:dual}.

\end{proof}

\subsection{Proof of Theorem~\ref{thm:main}}

By Proposition~\ref{prop:golfing}, an approximate dual certificate exists with probability at least
$1-\frac{1}{2}n^{-\omega}$ if
\begin{equation*}
L\gtrsim_{M,\nu}\,\omega\,\log n.
\end{equation*}
This condition also gives
\begin{equation*}
2n\exp\lk-c\frac{\nu^4L}{M^8}\rk\le\frac{1}{2}n^{-\omega},
\end{equation*}
so the robust injectivity in Proposition~\ref{prop:robust} holds with probability at least $1-\frac{1}{2}n^{-\omega}$. 
By a union bound, these two events hold simultaneously with probability at least $1-n^{-\omega}$, and Proposition~\ref{prop:recovery} yields exact recovery on their intersection.

\section{Numerical Experiments}\label{sec:numerics}

We conduct two complementary experiments.  
The first illustrates empirical behavior consistent with our theoretical prediction by comparing fixed and logarithmic mask budgets in the noiseless setting,
while the second probes the empirical robustness of PhaseLift under photon-limited observations. 
Throughout, we assume $\lV\x_0\rV_2=1$ and consider the coordinate and flat cases:
\begin{equation*}
    \x_0^{\mathrm{coord}}=\ve_1,
    \qquad
    \x_0^{\mathrm{flat}}
    =\frac{1}{\sqrt{n}}\lk 1,\ldots,1\rk^\top.
\end{equation*}
%The coordinate signal is the sparse extremal example underlying the erasure-mask lower bound, whereas the flat signal provides a dense contrast.
Masks are drawn independently from either the octanary ensemble in~\eqref{eq:octanary_mask} or the erasure ensemble in~\eqref{eq:erasure_mask}. 
We use even dimensions for octanary masks and odd dimensions for erasure masks, the latter matching the setting in which the real-valued mask condition is admissible.  

\paragraph{Sampling rate scaling.}
We compare the fixed mask budget $L=4$ with the logarithmic budget $L=\lceil4\log n\rceil$.  
Since each mask yields $n$ intensities, these budgets correspond to $m=4n$ and $m=n\lceil4\log n\rceil$ scalar measurements.  
For the octanary ensemble we test $n=8,16,\ldots,96$, and for the erasure ensemble we test $n=7,15,\ldots,95$.  
At each dimension we generate $20$ independent mask sequences of length $\lceil4\log n\rceil$.  
The first four masks constitute the fixed budget, so the two mask budgets are paired within each trial, and the same mask sequence is used for both test signals.

We solve the feasibility program~\eqref{eq:PhaseLift} with CVXPY~1.6.0 and SCS~3.2.7, setting the tolerance to $3\times10^{-6}$ and the maximum number of iterations to $10^3$.
All $1920$ solves return finite estimates. 
A trial is counted as successful only if
\begin{equation}\label{eq:numerical_success}
\frac{\lV\widehat{\X}-\X_0\rV_F}{\lV\X_0\rV_F}
\le 10^{-2},\quad
\frac{\lV\mA\lk\widehat{\X}\rk-\mA\lk\X_0\rk\rV_2}{\max\{\lV\mA(\X_0)\rV_2,1\}}
\le 10^{-4}.
\end{equation}
In all runs, the relative measurement residual is below $1.3\times10^{-6}$, so the success classification is determined by the reconstruction error criterion in~\eqref{eq:numerical_success}.

Figure~\ref{fig:octanary1} and Figure~\ref{fig:erasure1} show the empirical recovery rates at each tested dimension.  
Under the logarithmic budget, PhaseLift succeeds in $219$ of $240$ coordinate signal trials and all $240$ flat signal trials for the octanary masks. For the erasure masks, the corresponding success counts are $239$ of $240$ for both signal types.  
By contrast, under the fixed budget, PhaseLift succeeds in only $21$ of $240$ coordinate signal trials and $61$ of $240$ flat signal trials for the octanary masks, and in $42$ of $240$ and $20$ of $240$ trials, respectively, for the erasure masks.  
This pronounced separation is consistent with the sufficient sampling rate $m=\mathcal{O}\lk n\log n\rk$ established in Theorem~\ref{thm:main}.

\begin{figure}[htbp]
\centering
\includegraphics[width=\textwidth]{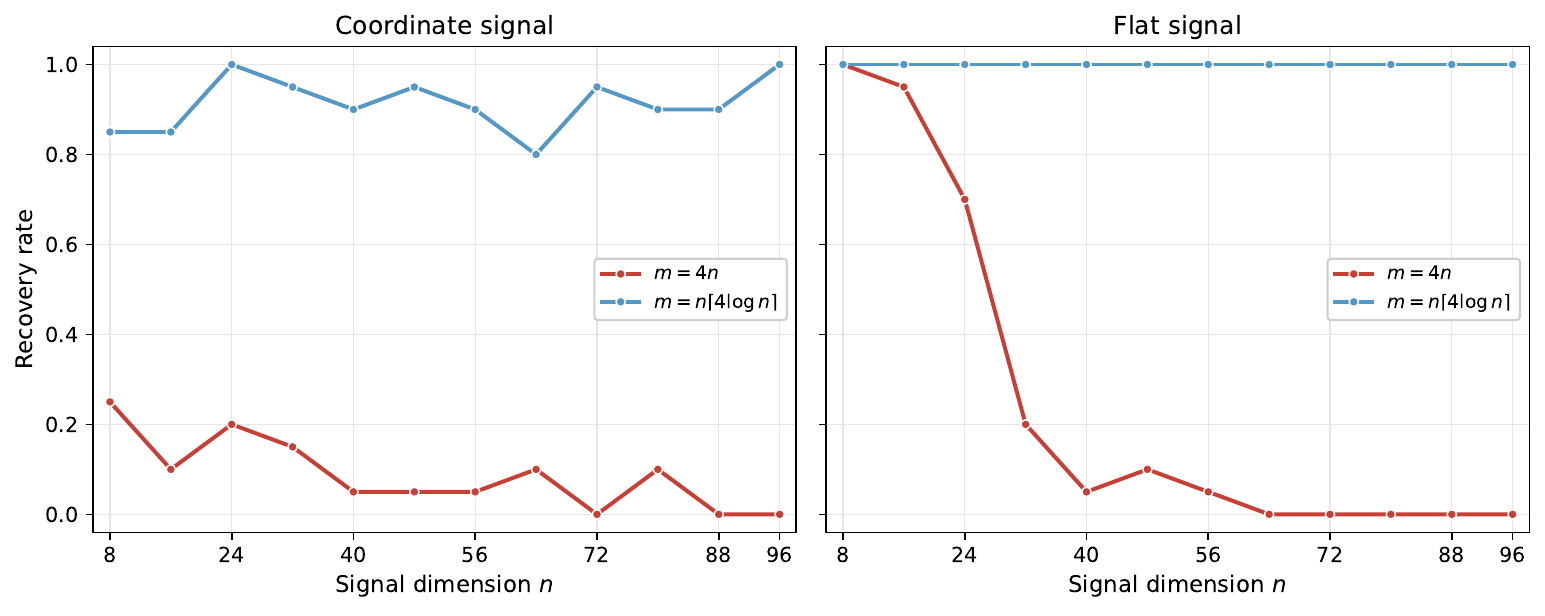}
\caption{Empirical recovery rates with octanary masks.}
\label{fig:octanary1}
\end{figure}

\begin{figure}[htbp]
\centering
\includegraphics[width=\textwidth]{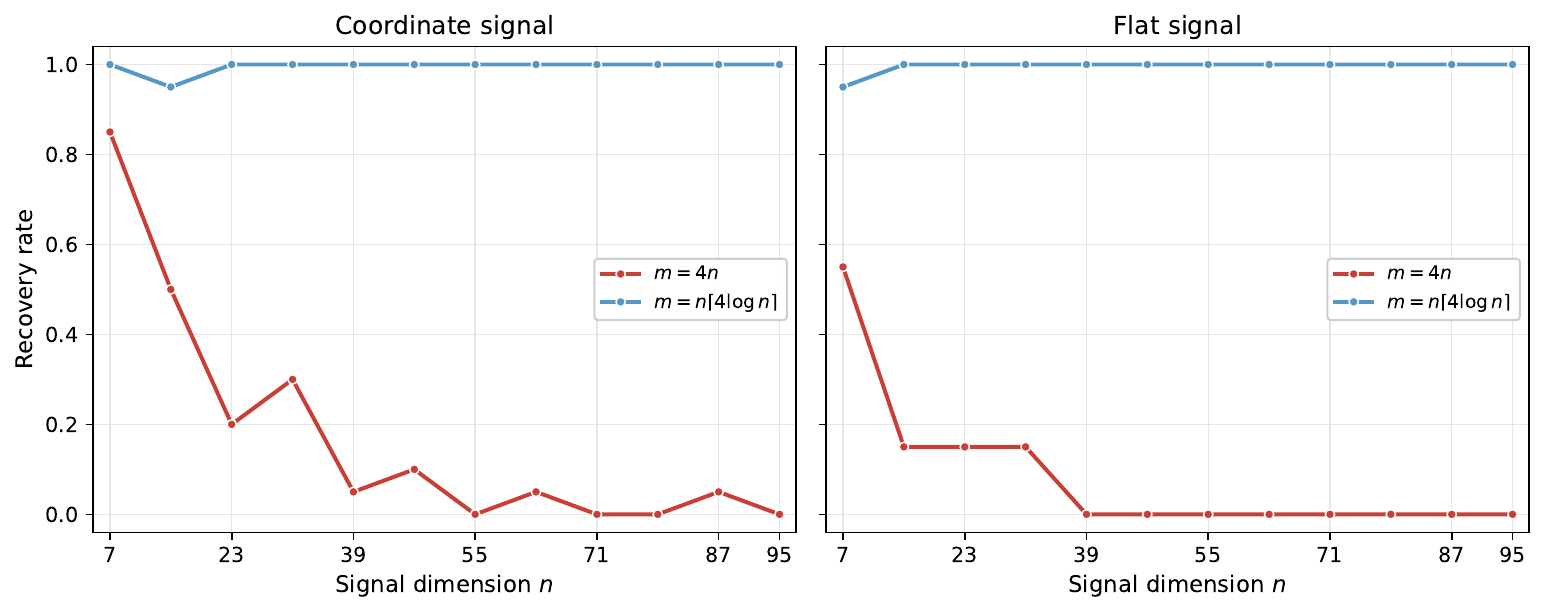}
\caption{Empirical recovery rates with erasure masks.}
\label{fig:erasure1}
\end{figure}

\paragraph{Poisson noise robustness.}
To examine a standard photon counting model beyond the noiseless setting, conditionally on the sampled masks, we generate independent Poisson observations according to
\begin{equation}\label{eq:poisson_model}
Z_{k,\ell}\overset{\mathrm{ind}}{\sim}
   \operatorname{Poisson}\lk\kappa y_{k,\ell}\rk,\qquad1\le k\le n,\quad 1\le \ell\le L,
\end{equation} 
with no background counts.
Here, $\kappa>0$ denotes the exposure level, so that $\kappa y_{k,\ell}$ is the expected photon count associated with intensity $y_{k,\ell}$; in particular, $\kappa$ is the expected count at unit intensity.
For $u,v\ge0$, define the generalized Kullback--Leibler divergence by
\begin{equation*}
D_{\mathrm{KL}}\lk u\,\middle\|\, v\rk := u\log\lk\frac{u}{v}\rk-u+v,
\end{equation*}
with the standard extended value conventions at zero.
We estimate $\X_0$ by solving the convex program
\begin{equation}\label{eq:poisson_phaselift}
\widehat{\X}_{\kappa}
\in
\mathop{\arg\min}_{\Z\succeq0,\ \tr\lk\Z\rk=1,\ \mA\lk\Z\rk\ge\pmb{0}}
\sum_{\ell=1}^{L}\sum_{k=1}^{n}
D_{\mathrm{KL}}\lk\frac{Z_{k,\ell}}{\kappa}\,\middle\|\,\tr\lk\Fm_{k,\ell}\Z\rk\rk.
\end{equation}
Up to an additive term independent of $\Z$ and multiplication by the positive scalar $\kappa$, this objective coincides with the Poisson negative log likelihood and therefore has the same minimizers.
%This normalization keeps the exponential cone variables at unit scale. 
We run SCS with a tolerance of $10^{-5}$ and a maximum of $5\times10^3$ iterations. 
All solver runs return finite matrix estimates, and we include every run in the reported summaries rather than filtering by solver status.
%Among the returned estimates, the largest absolute trace constraint violation is $3.1\times10^{-7}$, and the smallest eigenvalue is $-2.5\times10^{-4}$.

We fix $L=\lceil4\log n\rceil$ and consider dimensions $n\in\{16,32\}$ for octanary masks and $n\in\{15,31\}$ for erasure masks.  
The exposure levels are
\begin{equation*}
    \kappa\in\left\{1,3,10,30,100,300,1000\right\}.
\end{equation*}
For each combination of mask ensemble, dimension, and signal type, we perform $10$ independent trials, each using a fresh mask sequence and an independent Poisson process.
Within each trial, the masks are fixed, and the observations across exposure levels are coupled through Poisson increments, preserving the correct marginal distribution at each $\kappa$ while enabling paired comparisons.
We report the median relative matrix error
\begin{equation*}
    \frac{\lV\widehat{\X}_{\kappa}-\X_0\rV_F}{\lV\X_0\rV_F}.
\end{equation*}
The shaded regions in Figure~\ref{fig:octanary2} and Figure~\ref{fig:erasure2} indicate the interquartile range across the $10$ trials.

\begin{figure}[htbp]
\centering
\includegraphics[width=\textwidth]{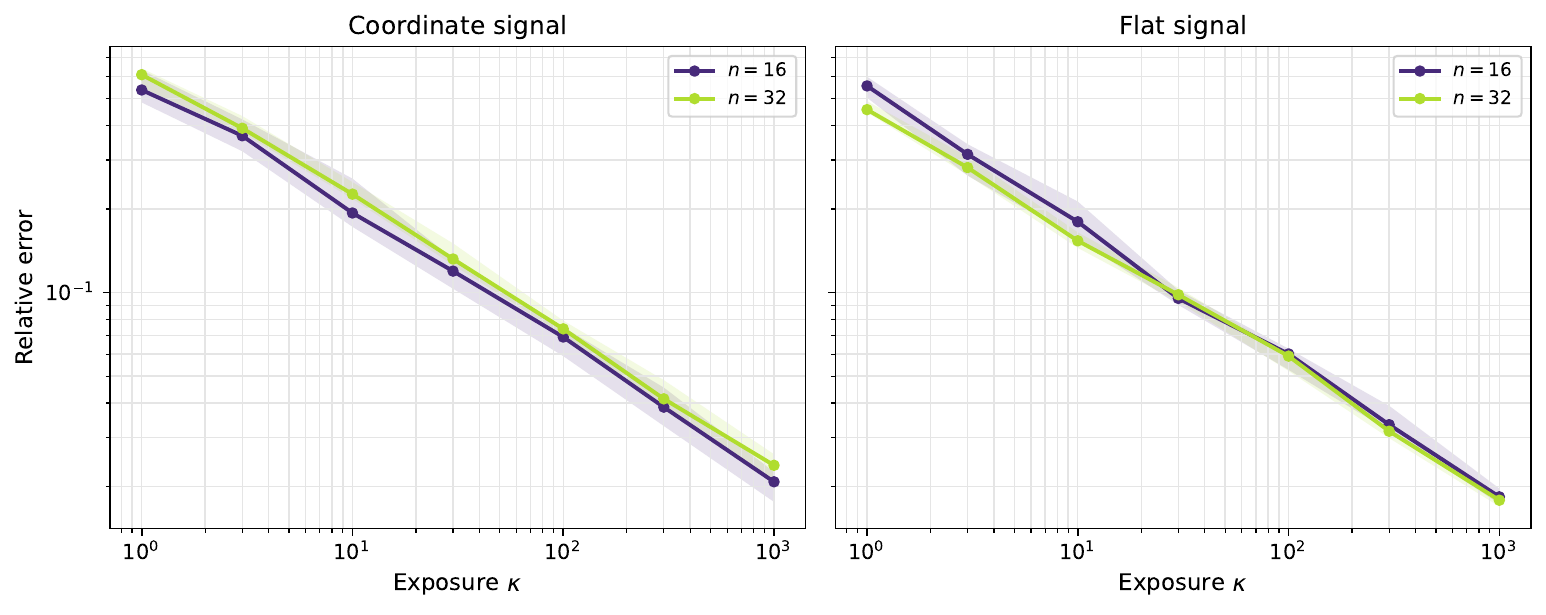}
\caption{Poisson-noise robustness with octanary masks.}
\label{fig:octanary2}
\end{figure}

\begin{figure}[htbp]
\centering
\includegraphics[width=\textwidth]{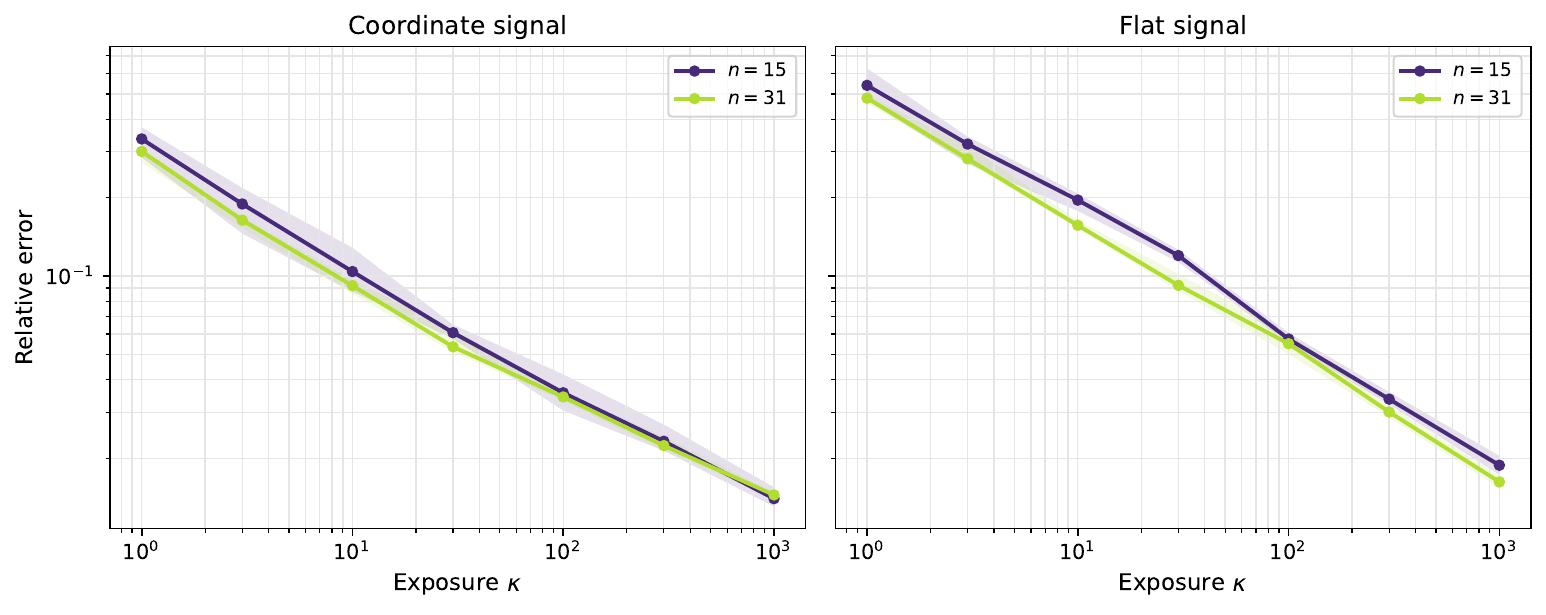}
\caption{Poisson-noise robustness with erasure masks.}
\label{fig:erasure2}
\end{figure}

Across all combinations of mask ensemble, dimension, and signal type, least-squares fits of the log-median error against $\log\kappa$ yield slopes between $-0.49$ and $-0.44$. 
These slopes are close to the canonical $\kappa^{-1/2}$ fluctuation scale of the normalized Poisson counts $Z_{k,\ell}/\kappa$. 
The median error decreases from between $0.30$ and $0.61$ at $\kappa=1$ to between $1.4\times10^{-2}$ and $2.4\times10^{-2}$ at $\kappa=1000$, with comparable decay rates for both mask ensembles and signal types.
Taken together, these findings provide empirical evidence that PhaseLift is robust to Poisson noise, while establishing corresponding theoretical guarantees remains an interesting direction for future work.

\section{Discussion}

This work establishes the optimal mask complexity, and hence the optimal sampling rate, for PhaseLift under the random mask model for coded diffraction patterns. 
For any unknown signal in $\C^n$ and every $\omega\ge1$, exact recovery is achieved from $\mathcal{O}\lk\omega\log n\rk$ random masks with failure probability at most $n^{-\omega}$. 
For the erasure mask ensemble, this dependence on both the dimension and the failure probability matches the corresponding lower bound, yielding the optimal total sampling rate.

%The proof relies on a refined golfing scheme that combines adaptive mask allocation with a dimension independent truncation level.
Several directions merit further investigation. 
First, our result is nonuniform, and establishing a uniform guarantee for the simultaneous recovery of all signals from a single realization of the random masks remains an important direction. 
Second, it would be interesting to determine whether other computationally tractable methods, particularly nonconvex algorithms, can attain the same optimal mask complexity with global exact recovery from a fixed collection of masks.
Finally, the numerical results suggest robustness to Poisson noise, motivating nonasymptotic stability guarantees that quantify the reconstruction error under photon-limited measurements.

\section*{Acknowledgments}

We would like to thank Huiping Li for helpful discussions during the early stages of this work.
This work was supported by the National Natural Science Foundation of China under Grant No.~U21A20426.

\appendix

\section{Mask Lower Bound}\label{ap:flat} 

The logarithmic lower bound for coordinate signals under the erasure mask ensemble~\cite[Lemma~19]{gross2017improved} is not specific to such highly localized signals.
We show that the same obstruction also arises for the flat signal (and, in fact, for a broader class of signals), for which $\Omega\lk\log n\rk$ masks are necessary.
Moreover, achieving success probability at least $1-n^{-\omega}$ requires $\Omega\lk\omega\log n\rk$ masks, showing that the dependence on the failure probability is also optimal up to constants.

\begin{proposition} \label{prop:lower_bound}
Let $n\ge2$, and suppose that the random masks are drawn independently according to the erasure mask ensemble in~\eqref{eq:erasure_mask}.
Consider the flat signal $\x_0=\frac{1}{\sqrt{n}}\lk 1,\ldots,1\rk^\top$.
Let $\mathsf{Inj}\lk\x_0\rk$ denote the event that $\x_0$ is uniquely determined, up to a global phase, by the CDP measurements in~\eqref{eq:cdp0}.
Then if $\mathbb{P}\lk\mathsf{Inj}\lk\x_0\rk\rk \ge1-\delta$ for some $\delta\in\lk0,1\rk$, then necessarily
\begin{equation}\label{eq:lower_bound}
        L
        \ge\log_2\lk\frac{n}{-\log\lk1-\delta\rk}\rk.
\end{equation}
Consequently, unique recovery of the flat signal with high probability requires $L=\Omega\lk\log n\rk$ masks, regardless of the recovery algorithm.
\end{proposition}

\begin{remark}\label{re:flat}
Taking $\delta=n^{-\omega}$ with $\omega\ge1$, Proposition~\ref{prop:lower_bound} yields
\begin{equation*}
L
\ge \log_2\lk\frac{n}{2n^{-\omega}}\rk
= \lk 1+\omega\rk\log_2 n-1.
\end{equation*}
Hence, recovering the flat signal with probability at least $1-n^{-\omega}$ requires $L=\Omega\lk\omega\log n\rk$ masks. 
Thus, for the erasure mask ensemble, the mask complexity $L=\mathcal{O}\lk\omega\log n\rk$ in Theorem~\ref{thm:main} is optimal when a success probability of at least $1-n^{-\omega}$ is required.
\end{remark}

\begin{proof}
For each $1\le j\le n$, define the event
\begin{equation*}
        \mathcal{E}_j:= \left\{d_{\ell,j}=0\text{for every}1\le\ell\le L\right\}.
\end{equation*}
The events $\mathcal{E}_1,\ldots,\mathcal{E}_n$ are independent, and $\mathbb{P}\lk \mathcal{E}_j\rk=2^{-L}$.

Suppose that $\mathcal{E}_{j_0}$ occurs for some $j_0$. 
Define
$\x_1 :=\x_0-\frac{2}{\sqrt n}\ve_{j_0}$.
Thus, $\x_1$ is obtained from $\x_0$ by changing the sign of its $j_0$-th coordinate, while $\x_1\neq\alpha\x_0$ for every $\alpha\in\C$ with $\lv\alpha\rv=1$. 
Since $\mathcal{E}_{j_0}$ occurs, $\x_1$ and $\x_0$ produce identical coded diffraction measurements.
Thus, the occurrence of any $\mathcal{E}_j$ prevents unique recovery of $\x_0$.
It follows that $ \mathsf{Inj}\lk\x_0\rk\subset\bigcap_{j=1}^n \mathcal{E}_j^c$.
By independence,
\begin{equation}
\mathbb{P}\lk\mathsf{Inj}\lk\x_0\rk\rk
\le\mathbb{P}\lk\bigcap_{j=1}^n \mathcal{E}_j^c\rk
=\lk1-2^{-L}\rk^n
\le\exp\lk-n2^{-L}\rk.
\end{equation}
Finally, if the success probability is at least $1-\delta$, then taking logarithms and rearranging yields~\eqref{eq:lower_bound}.
\end{proof}

\section{Weighted Geometric Tail Bound}\label{ap:geo}

The following lemma controls a weighted sum of conditionally geometric waiting times and will be used to bound the total number of masks consumed by the adaptive golfing construction in Section~\ref{subsec:golfing}.

\begin{lemma}\label{lm:geometric}
Let $m_0,\ldots,m_{r-1}>0$ be deterministic numbers, and set $B:=\sum_{j=0}^{r-1}m_j $.
Let $N_0,\ldots,N_{r-1}$ be positive integer-valued random variables generated sequentially. 
Assume that, there exist deterministic numbers $p_0,\ldots,p_{r-1}\in[0,1)$ such that, for every $0\le j\le r-1$ and every integer $l\ge0$,
\begin{equation}\label{eq:conditional}
        \mathbb{P}\lk N_j>l \,\middle|\, N_0,\ldots,N_{j-1} \rk
        \le
         p_j^l,
\end{equation}
where, for $j=0$, the conditional probability is interpreted as an unconditional probability.
Then, for every $t>0$ and every $\theta>0$ satisfying
\begin{equation*}
        p_j e^{\theta m_j}<1,\qquad 0\le j\le r-1,
\end{equation*}
one has
\begin{equation}\label{eq:tail1}
\mathbb{P}\lk\sum_{j=0}^{r-1}m_j N_j > t \rk
\le
e^{-\theta t}\prod_{j=0}^{r-1}\frac{\lk 1-p_j\rk e^{\theta m_j}} {1-p_je^{\theta m_j}}.
\end{equation}
\end{lemma}

\begin{proof}
Fix $0\le j\le r-1$ and set $u_j:=\theta m_j$.
For every positive integer-valued random variable $N$, the tail-sum identity gives
\begin{equation*}
e^{u_jN}=
e^{u_j}+
\sum_{l=1}^{\infty} \lk e^{u_j\lk l+1\rk}-e^{u_jl} \rk\ind_{\{N>l\}}.
\end{equation*}
Consequently, by conditional expectation and~\eqref{eq:conditional},
\begin{align*}
\E\lk e^{u_jN_j}\,\middle|\,N_0,\ldots,N_{j-1} \rk
&=
e^{u_j}
+\sum_{l=1}^{\infty} \lk e^{u_j\lk l+1\rk}-e^{u_jl}\rk
\mathbb{P}\lk N_j>l \,\middle|\,  N_0,\ldots,N_{j-1} \rk\\
&\le
e^{u_j}
+\sum_{l=1}^{\infty} \lk e^{u_j(l+1)}-e^{u_jl} \rk p_j^l.
\end{align*}
Since $p_je^{u_j}<1$, the geometric series converges, and hence
\begin{align*}
\E\lk  e^{u_jN_j}\,\middle|\, N_0,\ldots,N_{j-1} \rk
&\le e^{u_j}+\lk e^{u_j}-1\rk \sum_{l=1}^{\infty} \lk p_je^{u_j}\rk^l\\
&= \frac{\lk 1-p_j\rk e^{u_j}}{1-p_je^{u_j}}.
\end{align*}

Applying~\eqref{eq:conditional} successively and conditioning on $N_0,\ldots,N_{j-1}$ at each stage, we obtain
\begin{equation*}
\E\exp\lk \theta\sum_{j=0}^{r-1}m_jN_j\rk
\le
\prod_{j=0}^{r-1}
\frac{\lk 1-p_j\rk e^{\theta m_j}}{1-p_je^{\theta m_j}}.
\end{equation*}
Therefore, by Markov's inequality,
\begin{align*}
\mathbb{P}\lk \sum_{j=0}^{r-1}m_jN_j>t\rk
&\le
e^{-\theta t}\E\exp\lk\theta\sum_{j=0}^{r-1}m_jN_j\rk\\
&\le
e^{-\theta t}\prod_{j=0}^{r-1}\frac{\lk 1-p_j\rk e^{\theta m_j}}{1-p_je^{\theta m_j}},
\end{align*}
which proves~\eqref{eq:tail1}.

\end{proof}

\bibliographystyle{plain}

\bibliography{ref}

\end{document}